\documentclass{LMCS}

\def\dOi{10(4:15)2014}
\lmcsheading%
{\dOi}
{1--27}
{}
{}
{Jan.~27, 2014}
{Dec.~23, 2014}
{}

\ACMCCS{[{\bf Theory of Computation}]: Randomness, geometry and
  discrete structures---Pseudorandomness and derandomization;
  Computational complexity and cryptography }

\usepackage{hyperref}
\usepackage{amssymb}

\newcommand{\pz}{$\Pi^0_1$\ }

\newcommand{\X}{{\mathcal X}}

\newcommand{\N}{{\mathbb N}}
\newcommand{\Q}{{\mathbb Q}}

\newcommand{\sg}{\sigma}

\newcommand{\fr}{{^{\frown}}}

\newcommand{\res}{\upharpoonright}
\newcommand{\ml}{Martin-L\"{o}f\ }

\newcommand{\Ptime}{\mbox{\rm P}}
\newcommand{\PTIME}{\mbox{\rm PTIME}}
\newcommand{\EXPTIME}{\mbox{\rm EXPTIME}}
\newcommand{\PSPACE}{\mbox{\rm PSPACE}}
\newcommand{\FPSPACE}{\mbox{\rm FPSPACE}}
\newcommand{\FPSPACEplus}{{\FPSPACE^+}}
\newcommand{\EXPSPACE}{\mbox{\rm EXPSPACE}}
\newcommand{\DSPACE}{\mbox{\rm DSPACE}}
\newcommand{\BP}{\mbox{\rm BP}}
\newcommand{\BPS}{\mbox{\rm BPS}}

\newtheorem{theorem}{Theorem}[section]
\newtheorem{definition}[theorem]{Definition}
\newtheorem{lemma}[theorem]{Lemma}

\newtheorem{corollary}[theorem]{Corollary}

\newtheorem{proposition}[theorem]{Proposition}

\begin{document}

\pagestyle{headings}

\title{Sub-computable Bounded  Randomness}

\author[S.~Buss]{Sam Buss\rsuper a}
\address{{\lsuper{a,c}}Department of Mathematics,
University of California, San Diego,
La Jolla, CA 92093-0112}
\email{\{sbuss,jremmel\}@math.ucsd.edu}
\thanks{{\lsuper a}Supported in part by
NSF grants DMS-1101228 and CCR-1213151
and Simons Foundation award 306202.}

\author[D.~Cenzer]{Douglas Cenzer\rsuper b}
\address{{\lsuper b}Deptartment of Mathematics,
University of Florida,
Gainesville, FL 32611-8105}
\email{cenzer@ufl.edu}
\thanks{{\lsuper b}Supported in part by NSF\ grant DMS-652372.}

\author[J.~B.~Remmel]{Jeffrey B. Remmel\rsuper c}


\titlecomment{This
paper is an expanded version, with new results, of
a paper \cite{CenzerRemmel:BPrandomness} by two of the
authors}

\begin{abstract} This paper defines a new notion of bounded
computable randomness
for certain classes of sub-computable functions which lack
a universal machine.
In particular, we define
such versions of randomness for  primitive recursive functions and for
PSPACE functions. These new notions are robust in that
there are equivalent formulations in terms of
(1) Martin-L\"of tests,
(2) Kolmogorov complexity, and (3) martingales.  We
show these notions can be equivalently defined with
prefix-free Kolmogorov complexity.  We prove
that one direction of van Lambalgen's theorem holds for relative computability,
but the other direction fails.
We discuss statistical properties of these notions
of randomness.
\end{abstract}

\maketitle

\section{Introduction}\label{sec:introduction}

The study of algorithmic randomness has flourished over the past
century.  The main topic of study in this paper
 is the randomness of a single real
number which, for our purposes, can be thought of as
an infinite sequence $X = (X(0),X(1),\dots)$
from $\{0,1\}^\omega$.  Many interesting notions of algorithmic randomness for
real numbers have been investigated in recent years.  The most
well-studied notion, \ml randomness \cite{ML66} or 1-randomness, is
usually defined in terms of measure. Thus we say a
 real~$X$ is 1-random if it is
\emph{typical}, that is, $X$ does not belong to any effective set of measure
zero  in the sense of Martin-L\"{o}f \cite{ML66}.  A second definition of
1-randomness may be given in terms of information content: $X$~is
1-random if it is \emph{incompressible}, that is, the initial segments
$(X(0),X(1),\dots,X(n))$ have high \hbox{Kolmogorov}~\cite{K65} or
Levin-Chaitin~\cite{Chaitin66,Lev73} complexity.  A third definition
may be given in terms of martingales:  $X$~is 1-random if it is
{\em unpredictable}, that is, there is no effective martingale for which one
can obtain unbounded capital by betting on the values of $X$~\cite{V39}.
These three
versions have been shown by Schnorr~\cite{S71} to be equivalent.  This
demonstrates the robustness of the concept of \ml randomness.
Other interesting notions of algorithmic randomness have been studied, although in some
cases, formulations are only given for one or perhaps two
versions. For a thorough study of the area of algorithmic randomness,
the reader is directed to three excellent recently published books:
Downey and Hirschfeldt~\cite{DH11}, Nies~\cite{Niesbook} and Li and
Vitanyi~\cite{LV08}.

This paper presents a notion of bounded randomness
for some classes of sub-computable functions which do not
have access to a universal machine for that class.
We begin in Section~\ref{sec:BPrandomness} by stating
our definitions for the class of
primitive recursive functions and define a new notion of
\emph{bounded primitive recursive randomness} ($\BP$ randomness).
We show that there are three equivalent definitions of $\BP$ randomness:
one in terms of measure, one in terms of compressibility, and one in terms of martingales.   For
measure, a {\em bounded primitive recursive test} will be a primitive
recursive sequence of clopen sets $(U_n)_{n \geq 0}$ such that
$U_n$ has measure $\leq 2^{-n}$.  We define
$X$ to be {\em $\BP$ random} if it does not belong to $\bigcap_{n \geq 0} U_n$ for any
such test. For compressibility, we say that
$X$ is {\em $\BP$ compressed} by a primitive
recursive machine~$M$ if there is a primitive recursive function~$f$
such that $C_M(X \res f(c)) \leq f(c) - c$ for all~$c$, where
$C_M$ is a primitive recursive analogue of
Kolmogorov complexity. We show that
$X$ is $\BP$ random if and only if $X$ is not $\BP$ compressed by any
primitive recursive machine. We will also consider process machines
and the resulting notion of process complexity as studied recently by
Day~\cite{Day12}.

For martingales, we say that a primitive recursive martingale~$d$ {\em succeeds} on a
sequence~$X$ if there is a primitive recursive function~$f$ such that
$d(X \res f(n)) \geq 2^n$ for each~$n$. Thus $d$ makes us rich betting
on~$X$, and $f$~tells us how fast this happens.  We show
that $X$ is $\BP$ random if and only if there is no primitive recursive
martingale which succeeds on~$X$.

As we shall see, the presence of the auxiliary functions $f$ in
the definition of $X$ being $\BP$-compressed by a primitive
recursive machine and in the definition of when a primitive recursive
martingale succeeds on a sequence $X$ is key to our ability to prove that
the three definitions of $\BP$ randomness are all equivalent.
These definitions can be adapted to define a notion of
bounded randomness for other classes of sub-computable functions and
if such a  class satisfies certain closure conditions, then our
proofs of the equivalence of the three notions of primitive recursive
randomness can be adapted to prove similar results for that
class of sub-computable functions.  However some care has
to be taken in the definitions of the analogue of the auxiliary functions
$f$ described above.  As an example, in
the last part of the paper, we define a notion of
bounded polynomial space randomness, called $\BPS$ randomness.
We show that the results
obtained for $\BP$ randomness carry over, after suitable modifications,
to $\BPS$ randomness.

One motivation for the definition of $\BPS$ randomness is that it gives
a natural notion of randomness for polynomial space which has
equivalent formulations in terms of compressibility, in terms of
statistical tests, and in terms of martingales. Moreover,
our approach suggests a general alternative approach to defining randomness
via martingales for various classes of sub-computable functions by
adding the requirement that the rate of success of a strategy is measured
by a function in that class.  As we discuss below,
our notion of bounded randomness differs
from the well-studied notions of resource-bounded randomness which have
been developed over the past twenty-five years.  Finally, our general
approach to defining randomness on sub-computable classes of
functions should be useful for the study of resource-bounded trees
and effectively closed sets.

The terms \emph{bounded randomness} and \emph{finite randomness} are
sometimes used to refer to versions of randomness given by tests in
which the c.e. open sets are in fact clopen.
The term ``finite'' comes
from the fact that any clopen set~$U$ is the finite union of intervals
$U = [\sigma_1] \cup \dots\cup [\sigma_k]$.
Our notion
of $\BP$ randomness is ``bounded'' in this sense.
Kurtz
randomness \cite{Kurtz}, also referred to as weak randomness,
falls into this category. A real~$X$ is Kurtz
random if it does not belong to any \pz{}~class~$Q$ of measure zero. But
any \pz{}~class may be effectively expressed as a decreasing
intersection of clopen sets
$Q = \bigcap_n Q_n$ where the clopen sets $Q_n$ are unions
of intervals of length~$n$. If $\mu(Q) =0$, it is easy to find a
subsequence $U_i = Q_{n_i}$ with $\mu(U_i) \leq 2^{-i}$ and thus
$(U_n)_{n \geq 0}$ is a
bounded \ml test.   From this it follows immediately
that if our definition of $\BP$ random is modified to use
computable functions instead of primitive recursive functions,
then it becomes equivalent to Kurtz randomness.
Kurtz randomness
has been well-studied, and is already known to have equivalent
definitions in terms of \ml randomness, Kolmogorov
compressibility, and effective martingales; for this see Wang~\cite{Wang00},
Downey, Griffiths and Reid~\cite{DGR}, and
Bienvenu and Merkle~\cite{BM07}.
Uniform relative Kurtz randomness is studied by
Kihara and Miyabe \cite{KM13}, who proved van Lambalgen's
theorem holds for their definition.
Another special type of bounded randomness was recently
studied by Brodhead, Downey and Ng~\cite{BDN12}.

As shown by Wang~\cite{Wang00}, Kurtz random reals may not be
stochastic in the sense of Church~\cite{Church}. For example, it need not be the
case that the number of occurrences of~0's in a Kurtz
random sequence~$X$ tends to $1/2$ in the limit. This can
also happen for our $\BP$ random reals. Indeed,
we construct a
recursive real which is $\BP$ random but not stochastic. However,
in Section~\ref{sec:statisticaltests}, we show
that $\BP$ random sets do satisfy a weak
``infinitely often'' version of the stochastic
property.

A lot of work has been done on various notions of resource-bounded randomness.
One of the first approaches to resource-bounded randomness was via the
above-mentioned stochastic property of typical reals.  It is expected
that for a random real, the relative densities of 0's
and 1's  should be equal in the limit.  We identify a set $A$ of
natural numbers with its characteristic function and in those terms we
expect that $\lim_n (card(A \cap [[n]])/n = \frac12$ where
$[[n]] = \{0,1, \ldots, n-1\}$.  Levin~\cite{Lev73}
defined a notion of primitive randomness for a set $A$ to
mean that for every primitive recursive set~$B$,
the set $A \cap B$ is
stochastic relative to~$B$, and constructed a recursive set that is
primitive random. Di~Paola~\cite{DP69} studied similar notions of
randomness in the Kalmar hierarchy of elementary functions.
Wilber~\cite{Wilber83} defined a set~$A$ to be $\Ptime$ random if, for every
$\PTIME$ set~$B$, the sets $A$ and~$B$
agree on a set of density $\frac12$, and
constructed an exponential time computable $\Ptime$ random set.

The literature of computational complexity contains many papers on
random number generators and cryptography which examine various
notions of pseudorandomness.  For example, Blum and Micali~\cite{BM84}
gave a weak definition of pseudorandom sequences in which a randomly
generated sequence is said to be pseudorandom if it meets all $\PTIME$
statistical tests.  Ko~\cite{Ko86} gave definitions of randomness
with respect to polynomial time and polynomial space complexity
which are in the
tradition of algorithmic randomness as established by Levin, \ml,
Chaitin, and Schnorr. Ko's notion of polynomial space randomness
has equivalent formulations in terms
of tests and in terms of compressibility
(namely, the equivalent classes PSR1 and PSR2 in~\cite{Ko86}).
Ko also shows PSR1-/PSR2-random reals satisfy strong
stochastic properties.
The intuition behind Ko's polynomial space randomness
is very different from the motivation of our BPS randomness.
Indeed, Ko defines a notion of randomness based on the
intuition that no polynomial space test can ``reject the hypothesis
that $x$ is random on the significance level of $2^{-m}$'' \cite{Ko86}.
However, this is exactly the notion of polynomial
space test that we {\em are} adopting.
For instance, Ko's definition of PSR2 is based on
compressibility using a universal polynomial space machine~$M$
and states
that a real~$X$ is ($\PSPACE$) compressed if,
for every~$k$, there exists
infinitely many $n$ such that $K_M(X \res n) < n - (\log n)^k$,
where $K$ is a notion of Kolmogorov complexity.  (We have not
stated the definition fully: see~\cite{Ko86}
for the correct, precise definition.)
In contrast, our definition of BPS states, loosely speaking,
that $X$ is ($\PSPACE$) compressed provided there
is some~$k$ such that $K_M(X \res n^k) \le n^k-n$ for all sufficiently
large~$n$. (See
Section~\ref{sec:BPSrandomness} for the actual definition; in particular,
unlike Ko, we do not use a universal machine~$M$).
Clearly every real which is polynomial space random in Ko's sense is
also $\BPS$ random.  The converse does not hold, as BPS randomness
does not have as strong stochastic properties.
For this, see the discussion after Theorem~\ref{thm:nolimitBPS}.

Lutz~\cite{Lutz92}, building on Schnorr~\cite{S71},
defined an important notion of resource-bounded
randomness in terms of martingales. For the special
case of polynomial space, Lutz defined that a real~$X$
is $\PSPACE$ random
if there is no $\PSPACE$ martingale which succeeds on~$X$.
This differs from our definition of $\BPS$ randomness
in that there is no auxiliary function~$f$ which bounds the
time-to-success for the martingale.
We shall
see that this is a strictly stronger notion than $\BPS$ random.
Clearly, every $\PSPACE$ random real is automatically $\BPS$ random.
On the other hard, there
are $\BPS$ random reals which are not $\PSPACE$ random in the sense of
Lutz.  This will be an immediate consequence of our
Theorem~\ref{thm:nolimitBPS}, which states that $\BPS$ random reals
do not
enjoy as strong stochastic properties as are known to
hold for $\PSPACE$ random reals.
Lutz further defined that a set~$\X$ of reals
has $\PSPACE$ measure one if there is no
$\PSPACE$ martingale
which succeeds on every element of~$\X$. Then almost
every $\EXPSPACE$ real is random, and this can be used to study
properties of $\EXPSPACE$ reals by examining whether the set of $\EXPSPACE$
reals with the property has measure one. Buhrman and Longpre~\cite{BL02} gave a
rather complicated equivalent formulation of $\PSPACE$ randomness in
terms of compressibility.  Lutz's notion of complexity theoretic
randomness concept has been developed further by
\cite{AS94,AS95,AM97,Moser:Martingale}. Shen et al.~\cite{CSVV} have recently studied
on-line complexity and randomness.

Ville's theorem is a fundamental property of \ml random reals
and states that any effective subsequence of a random
sequence is also random. We prove an analogue
 of Ville's theorem for $\BP$ randomness.

Another fundamental
property for random reals is van Lambalgen's theorem, which
states that the join $A \oplus B$ of two random sets is random if and only if
$B$~is random and $A$~is random relative to~$B$.
We define a natural notion of {\em relative
$\BP$ randomness} which still has three equivalent formulations.
We prove one direction of the analogue of van Lambalgen's theorem,
showing that if $A$~is $\BP$ random relative to~$B$ and $B$~is $\BP$ random relative to~$A$,
then $A\oplus B$ is $\BP$ random.  However, we give a counterexample
for the converse, by showing there exist $A$ and~$B$ such that
$A$~is $\BP$ random relative
to~$B$ and $B$~is $\BP$ random relative to~$A$,
but $A\oplus B$ is not $\BP$ random.  This corrects
Theorem~6 of~\cite{CenzerRemmel:BPrandomness} which claimed that
both directions of van Lambalgen's theorem hold for $\BP$ and $\BPS$ randomness.

For the case of bounded $\BPS$ randomness,
we again present equivalent notions,
one in terms of compressibility, one in terms of measure,
and one in terms of martingales.  In this case, polynomial growth
rate functions are used to bound the rate of compressibility,
the size of the measure-based
tests, and the success rate of the martingales.
The compressibility definition is shown to be equivalent whether
defined in terms of prefix-free or non-prefix-free functions.
We prove that one direction of van Lambalgen's theorem holds for relative randomness
of bounded $\BPS$ randomness, whereas the other direction fails.
We also discuss stochastic properties
of bounded $\BPS$ random sequences.

The outline of this paper is as follows.
Section~\ref{sec:BPrandomness} studies $\BP$ randomness and
shows the equivalence of our three versions. We construct a
computable real which is $\BP$ random. We prove an analogue of Ville's
theorem for primitive recursive subsequences of $\BP$ random reals.
We also define a notion of relative randomness and discuss
the analogue of van Lambalgen's theorem.
Section~3 presents three equivalent notions of
bounded polynomial space
randomness (which we call ``$\BPS$ randomness'' to distinguish it from
earlier different definitions); and extends all the constructions of
Section~\ref{sec:BPrandomness} to bounded $\BPS$ randomness.
Section~4 states our conclusions and some
directions for further research.

A preliminary version of this paper was published in \cite{CenzerRemmel:BPrandomness}.
The present paper has the following improvements.
It includes all of the proofs for the case of $\BPS$ randomness which were omitted in the conference version.
It presents a new definition of $\BP$ randomness and $\BPS$ randomness
using process machines and the proof of their equivalence with the other definitions.
It proves the equivalence of $\BPS$ randomness and weak $\BPS$ randomness.
It gives a proof of the intuitive fact that prefix-freeness is not needed for the Kolmogorov complexity version of $\BP$ and $\BPS$ randomness.
It gives a proof that only one direction of van Lambalgen's theorem holds for $\BP$ and $\BPS$ randomness. This corrects
Theorem~6 of  \cite{CenzerRemmel:BPrandomness},
which claimed that both directions of van Lambalgen's Theorem
held for $\BP$ and $\BPS$ . It contains an improvement in the
construction of a $\BP$ random real which significantly lowers the complexity
from $2^{2^n}$ to $n^{\epsilon \log n}$.

\section{Bounded primitive recursive randomness}\label{sec:BPrandomness}

This section defines the three notions of primitive recursive
randomness, Kolmogorov $\BP$ randomness,
\ml $\BP$ randomness, and martingale $\BP$ randomness,
and proves their equivalence.
Hence we say that a real $X$ is $\BP$ random if
it satisfies these three definitions.
Subsection~\ref{sec:prefixfreeBP}
introduces
a notion of prefix-free $\BP$ randomness, and shows it is also
equivalent to $\BP$ randomness.   We also
prove a simple analogue of Ville's theorem
and discuss statistical tests and $\BP$ randomness.

\subsection{Three definitions for BP randomness}\label{sec:definitionsBP}

\subsubsection{Martin-L{\"o}f BP randomness}\label{sec:MartinLofBP}

We work with the usual alphabet $\Sigma = \{0,1\}$ and
the corresponding set $\{0,1\}^*$ of finite strings and the Cantor space
$\{0,1\}^\omega$ of infinite sequences, but our results hold for any finite
alphabet.
In this section, we study primitive recursive functions
$M: \Sigma^* \to \Sigma^*$.  We code finite strings as numbers in order
to define these primitive recursive functions using primitive
recursive coding and
decoding functions.
The code $c(\sg)$ of a finite sequence
$\sg =\sg_1 \cdots \sg_n \in \{0,1\}^*$
is just the natural number whose binary expansion is $1\sg_1\cdots \sg_n$.
Given a nonempty finite set $S=\{\sg^{(1)}, \ldots , \sg^{(k)}\}$ of
strings in $\{0,1\}^*$ such that
 $c(\sg^{(1)}) < \cdots < c(\sg^{(k)})$, the code $C(S)$ of~$S$ is defined
to be the natural number~$n$ whose ternary expansion is
$2c(\sg^{(1)})2 \cdots 2 c(\sg^{(k)})$.
We let $\emptyset$ denote the empty string.
The set of infinite $0$/$1$ sequences is
denoted $\{0,1\}^\omega$; these are also
called {\em reals}.
An $X\in \{0,1\}^\omega$ is also viewed as a
set, namely the set $\{i : X(i)=1\}$.

For any string $\sg \in \{0,1\}^*$,
$|\sigma|$~denotes the length of~$\sigma$.  For
$X \in \{0,1\}^* \cup \{0,1\}^\omega$,
we write $\sg \sqsubset X$ if $\sg$ is
an initial segment of $X$, and let
$[\sigma] = \{X \in \{0,1\}^\omega$ such that $\sg \sqsubset X\}$.
For a set $G$ of
strings in $\{0,1\}^*$, we let $[G] = \bigcup\{[\sg]: \sg \in G\}$.
We say a sequence $(U_n)_{n \in \N}$ of clopen
sets is a {\em primitive recursive sequence} if there is a
primitive recursive function~$f$ such that for all~$n$, $f(n)$ is a code
of a finite set $G_n = \{\sg_{1,n}, \ldots, \sg_{k(n),n}\}$ such that
$U_n = [G_n]$.

We define a {\em primitive recursive test} to be a primitive recursive
sequence $(U_n)_{n \geq 0}$ of clopen sets such that, for each $n$,
$\mu(U_n) \leq 2^{-n}$.
Without loss of generality, we may assume that there is a primitive
recursive function~$g$ such that $g(n)$ is a code of a finite set
$G_n = \{\tau_{1,n}, \ldots, \tau_{r(n),n}\}$
such that $U_n = [G_n]$ and
$|\tau_{1,n}|= \cdots = |\tau_{r(n),n}| = \ell(n)$
where $r$ and~$\ell$ are also primitive recursive functions.
It follows that there is a primitive recursive function~$m$ such that $m(n)$
codes the measure $\mu(U_n)$ as a dyadic rational.
Since the measures $\mu(U_n)$ may be computed
primitive recursively, one could equivalently
define a primitive recursive test to be a
a primitive recursive sequence $(V_n)_{n \geq 0}$ such that
$\lim_n \mu(V_n) = 0$ and there is a primitive recursive function~$f$
such that, for each~$p$, $\mu(V_{f(p)}) \leq 2^{-p}$.

Observe that $\bigcap_n U_n$ is a
\pz{}~class of measure zero, so any primitive recursive test is
a Kurtz test and hence also a Schnorr test.

\begin{definition}
An infinite sequence $X \in \{0,1\}^\omega$ is {\em \ml $\BP$ random}
if $X$ \emph{passes} every primitive recursive test, that is, for every
primitive recursive test $(U_n)_{n \geq 0}$, there is some~$n$ such
that $X \notin U_n$.
\end{definition}

By the remarks above, every Kurtz random real is  \ml $\BP$ random.

\begin{proposition} \label{prop1}
  $X$ is \ml $\BP$ random if and only if there is no
  primitive recursive sequence $(U_n)_{n \geq 0}$ of clopen sets with
  $\mu(U_n) = 2^{-n}$ such that $X \in \bigcap_n U_n$.
\end{proposition}

\begin{proof} The if direction is immediate. Now suppose that
there is a primitive recursive sequence $(V_n)_{n \geq 0}$ such that
$\mu(V_n) \leq 2^{-n}$ and $X \in \bigcap_n V_n$.
Let $V_n = \bigcup_{\sg \in G_n} [\sg]$
where $G_n \subseteq \{0,1\}^{\ell(n)}$
for some primitive recursive function $\ell(n)$, where $\ell(n)\geq n$
for all $n$.  Then
$\mu(V_n) = card(G_n)/2^{\ell(n)} \leq 2^{-n}$.
Define $H_n$ to be $G_n$ together with
$2^{\ell(n)-n} - card(G_n)$ additional strings of length $\ell(n)$ and
let $U_n = [H_n]$. Then for each~$n$,
$X \in U_n$ and $\mu(U_n) = 2^{-n}$.
\end{proof}


We also need the notion of a \emph{weak} primitive recursive
test. A \emph{weak} primitive recursive
test $(U_n)_{n\geq 0}$ is a primitive recursive test
for which there are
a primitive recursive sequence $(G_n)_{n \geq 0}$ and
a primitive recursive function $\ell$ such that
for each $n$,
$U_n = [G_n]$ and, for all $\tau \in G_n$,  $|\tau| = \ell(n)$  and
$\mu(U_{n+1} \cap [\tau]) \leq \frac12 \mu([\tau])$.

By definition, every weak primitive recursive test $(U_n)_{n \geq 0}$
is also a primitive recursive test.
Conversely,
we can convert a primitive recursive test
$(U_n)_{n \geq 0}$ into a weak primitive recursive test
as follows. First, we may assume that $U_{n+1} \subseteq U_n$ for each
$n$, since the sequence $W_n = \bigcap_{i \leq n} U_i$ is
also a primitive recursive test with $\mu(W_n) \leq \mu(U_n) \leq
2^{-n}$. Next suppose
$U_{n} = [\tau_{1,n}] \cup \cdots \cup [\tau_{k(n),n}]$
where there is a primitive recursive function $\ell$ such that
$|\tau_{i,n}|= \ell(n)$ for all~$i$, so
each interval $[\tau_{i,n}]$ has measure
exactly $2^{-\ell(n)}$.  The clopen set $U_{\ell(n)+1}$ has a
total measure $\leq 2^{-\ell(n)-1}$, so
$\mu(U_{\ell(n)+1} \cap [\tau_{i,n}]) \le \mu(U_{\ell(n)+1}) \le \frac12 \mu([\tau_{i,n}])$.
Then we can define
a primitive recursive weak test $(V_n)_{n \geq 0}$ as follows.
Let $h(0) = 0$ and let $V_0 = U_0$. Then let $h(1) = \ell(0)+1$ and $V_1 = U_{h(1)}$.
In general for $n >1$, we let $h(n+1) = \ell(h(n))+1$ and let
$V_{n+1} = U_{h(n+1)}$.
The sequence $V_0,V_1,\dots$ will be a weak primitive recursive
test. Since the sequence $(V_n)_{ n \geq 0}$ is a subsequence of the
original sequence $(U_m)_{m \geq 0}$, it follows that $\bigcap_n V_n
= \bigcap_n U_n$, so that $X$ passes the weak test $(V_n)_{n \geq 0}$
if and only if it passes the original test.

We have established the following.

\begin{proposition} \label{prop2}
$X$ is \ml $\BP$ random if and only
if it passes every weak primitive recursive test.
\end{proposition}

\subsubsection{Kolmogorov BP randomness}\label{sec:KolmogorovBP}

Let $M:\{0,1\}^* \rightarrow \{0,1\}^*$ be a primitive recursive
function. Let $C_M(\tau)$ be the length~$|\sigma|$
of the shortest string~$\sigma$ such that $M(\sigma) = \tau$,
that is, the length of the
shortest $M$-description of~$\tau$. Notice that we are using plain and
not prefix-free complexity. We say that $X$ is
{\em primitive recursively compressed}
by~$M$ if there exists a primitive recursive function $f$ such that,
for every $c \in \N$,
$C_M(X\res f(c)) \leq f(c)-c$.

\begin{definition} \label{def:KolmogorovRandomDefn}
An infinite sequence $X \in \{0,1\}^\omega$ is
\emph{Kolmogorov $\BP$ random} if it is not primitive recursively
compressed by any primitive recursive function $M:\{0,1\}^* \rightarrow \{0,1\}^*$.
\end{definition}

Our definition of primitive recursive compressibility
is a natural effective analogue
of the usual definition of Kolmogorov compressibility
which says that, for every $c \in \N$, there exists $n$ such that $C_M(X \res
n) \leq n - c$. That is, given the formulation ``for every $c$, there exists $n$'', there must
exist a function~$f$ which computes a value $n = f(c)$ for each input $c$. The complexity of
the function~$f$ calibrates the difficulty of computing a value of $n$ for an arbitrary
input $c$.  In the setting
of primitive recursive computability,
it is reasonable to require that $f$ be primitive recursive.
In addition, this makes the definition equivalent to the
definition given above for
$\BP$ \ml randomness.

The function $f$ can be assumed to be strictly increasing
without loss of generality; for this see the remark
after the proof of Theorem~\ref{bps}.

 Of course, one defines
Kolmogorov randomness in terms of prefix-free
complexity $K_M$ since there are no infinite Kolmogorov random
sequences for plain complexity.  The definition
of Kolmogorov $\BP$ random uses plain complexity since
every primitive recursive function is total so there are no
prefix-free machines.
However, Section~\ref{sec:prefixfreeBP}
considers a version of prefix-free
complexity for primitive recursive functions which
uses primitive recursive functions~$M$ such that
$M(\sigma)$ may diverge.  This will be done by introducing
a new symbol~$\infty$ as a possible output of $M(\sigma)$ to signify that
$M(\sigma)$ diverges.  It will not be hard to show that this makes no
difference.

We also consider process machines and the resulting notion of
process complexity.  A partial computable
function $M:\{0,1\}^*\rightarrow\{0,1\}^*$ is said to be a \emph{process machine}
if, whenever $\tau \sqsubset \tau'$ and $\tau,\tau' \in Dom(M)$, then
$M(\tau) \sqsubset M(\tau')$, so that $M$~is extension
preserving when defined.  A process machine~$M$
is a \emph{strict process machine}
if, whenever $\tau \sqsubset\tau' \in Dom(M)$,
then $\tau \in Dom(M)$, so the domain of~$M$ is
closed under prefixes.
A strict process
machine~$M$ is a \emph{quick process machine} if $M$~is total and
there is an order function~$h$ such that, for all $\tau \in \{0,1\}^*$,
$|M(\tau)| \geq h(|\tau|)$.
Recall that $h:\N\rightarrow \N$ is an {\em order function}
provided $h$~is non-decreasing and $\lim_n h(n) = \infty$.
A strict process machine~$M$
is called a \emph{quick process $\BP$ machine} if
the order function~$h$ is primitive recursive.
Finally, $X\in\{0,1\}^\omega$ is {\em quick process $\BP$ random}
provided $X$~is not primitively recursively compressed by any
quick process $\BP$ machine.

The definition of a process machine is due to
Levin and Zonkin~\cite{LZ70} and a
similar notion was defined by Schnorr~\cite{S73}.
Day~\cite{Day12} gives characterizations of
computable randomness, Schnorr randomness, and weak randomness
using quick process machines.
Theorem~\ref{bps} below shows
that $X$ is Kolmogorov $\BP$ random
if and only if $X$ is quick process $\BP$ random.

\subsubsection{Martingale BP randomness}\label{sec:martingaleBP}

A martingale is a function $d: \{0,1\}^* \to \Q \cap [0,\infty)$
such that $d(\emptyset)=1$ and, for all $\sigma \in \{0,1\}^*$,
$d(\sigma) = (d(\sigma \fr 0)+d(\sigma \fr 1)) / 2$.
Of course, any primitive recursive
martingale is also a computable martingale. We say that the martingale~$d$
{\em succeeds primitive recursively on~$X$} if there is a
primitive recursive function~$f$ such that, for all $n$,
$d(X \res f(n)) \geq 2^n$.
(Of course, we could replace $2^n$ here with any strictly increasing
primitive recursive function.)  In
general, a martingale $d$ is said to succeed on $X$ if $\limsup_n d(X
\res n) = \infty$, that is, for every $n$, there exists $m$ such that
$d(X \res m) \geq 2^n$. Thus our definition is an effectivization of
the usual definition with a primitive recursive function~$f$
which witnesses some point where $d$ will return $2^n$.
A  martingale could be thought of as a financial advisor who guarantees
the client's  eventual wealth; effective success of the martingale means that
the advisor can predict when the client will reach a given level of wealth.
Again, the requirement that $f$ be primitive recursive
is necessary to have the equivalence
with  $\BP$ \ml randomness.
\begin{definition}
$X$ is \emph{martingale $\BP$ random} if there is no
primitive recursive martingale which succeeds primitive recursively on~$X$.
\end{definition}
If $X$ is \emph{not}  martingale $\BP$ random, then there is a
computable martingale which succeeds primitive recursively on~$X$ and
thus certainly succeeds on~$X$, so $X$ is not computably random.
Hence every computably random real is also a martingale $\BP$ random
real.

The definition of martingale $\BP$ random real has
the following equivalent formulations.

\begin{proposition} \label{prop3} The following are equivalent.
\begin{enumerate}
\item $X$ is  martingale $\BP$ random.
\item There do not exist a primitive
  recursive martingale~$d$ and a primitive recursive function~$f$ such
  that, for every~$n$, there exists $m\leq f(n)$
  such that $d(X \res m) \geq 2^n$.
\item There do not exist a primitive recursive martingale~$d$
  and a primitive
  recursive function~$f$ such that $d(X \res m) \geq 2^n$ for
  all $n$ and all $m \geq f(n)$.
\end{enumerate}
\end{proposition}

\begin{proof}
The implications (1)\,$\Rightarrow$\,(3) and (2)\,$\Rightarrow$\,(1) are obvious.
The proof of (3)\,$\Rightarrow$\,(2)
uses the idea of a \emph{savings account} as formulated in
\cite{DH11,Niesbook}.  Let the martingale~$d$ and function~$f$
be as in (2). We shall  modify these to form a
new pair $d$ and~$f$ which satisfy~(3).
The intuition is that $d(\tau)$ is modified so that whenever $d(\tau)\ge 2^{n+1}$,
then one half of the working capital, as represented by~$d(\tau)$,
is transferred to a savings account.  That is, loosely speaking,
every time the capital is doubled, one half of it is placed into
the savings account.
Formally, we define $d^\prime(\tau)$ as follows.
Given $\tau\in \{0,1\}^*$, define
$\ell^\tau_n$ to equal the least value,
if any, such that $d(\tau\res \ell^\tau_n) \ge 2^{n+1}$.
Let $\tau_n$ equal $\tau \res \ell^\tau_n$, and let $n(\tau)$ be the
least value of~$n$ for which $\ell^\tau_n$ is undefined.
Thus, $\tau_{n(\tau)-1}\sqsubset \tau$, and $\tau_{n(\tau)}$
is undefined.
Note $n(\emptyset) = 0$.
The martingale~$d^\prime(\tau)$ is defined by
\[
d^\prime(\tau) ~= ~ \frac{d(\tau)}{2^{n(\tau)}} +
   \sum\limits_{i=0}^{n(\tau)-1} \frac{d(\tau_i)}{2^{i+1}}.
\]
As the reader can easily verify, $d^\prime$~is a martingale.
The intuition is that the summation on the righthand side
of the equation equals the capital placed in
the savings account, and
$d(\tau)/2^{n(\tau)}$ is the remaining working capital.

Since $d(\tau_i)\ge 2^{i+1}$, we have $d(\tau)\ge n$ whenever
$\ell^\tau_n$ is defined; in particular, this holds whenever
$d(\tau^\prime)\ge 2^{n+1}$ for some $\tau^\prime\sqsubset\tau$.
Finally, define $f^\prime(n)$ to equal $f(2^{n+1})$.  Let
$X\in\{0,1\}^\omega$  and $m\ge f^\prime(n)$.  Since
$m \ge f(2^{n+1})$, $d(X \res m)\ge 2^{2^{n+1}}$.
Therefore, $d^\prime(X \res m) \ge 2^n$.  This establishes
that the property of~(3) holds for $d^\prime$ and~$f^\prime$.
\end{proof}

\subsection{Equivalences}\label{sec:BPequivalences}

The main result of this section is the next
two theorems showing that the three versions of $\BP$ random
described above are equivalent.

\begin{theorem} \label{bps}
The following statements are equivalent for $X \in \{0,1\}^\omega$:
\begin{enumerate}
\item $X$ is \ml $\BP$ random.
\item $X$ is Kolmogorov $\BP$ random.
\item $X$ is quick process $\BP$ random.
\end{enumerate}
\end{theorem}

\begin{proof} The implication (2)\,$\Rightarrow$\,(3) is immediate.

{(1)\,$\Rightarrow$\,(2)}: Suppose (2) is false
and $X$ is not Kolmogorov $\BP$ random.  Then
there exist primitive recursive $M$ and~$f$ such that
$C_M(X \res f(c)) \leq f(c)-c$ for all $c \in \N$.

Let $U_c = \{Y\in\{0,1\}^\omega: C_M(Y \res f(c+1)) \leq f(c+1) - c-1\}$.
We claim this is a uniformly primitive recursive sequence of clopen
sets.  To see this, let
\[
G_c ~=~ \{M(\sigma): \sigma \in \{0,1\}^{\leq f(c+1) - c-1}\}
\cap \{0,1\}^{f(c+1)}.
\]
Then $U_c = [G_c]$, and this expresses
$(U_c)_{c \geq 0}$ as a primitive recursive sequence
of clopen sets.

We also claim that $\mu(U_c) \leq 2^{-c}$.  For this, fix~$c$ and let
$U_c = [\tau_1] \cup [\tau_2] \cup \dots \cup [\tau_k]$,
for distinct $\tau_i \in \{0,1\}^{f(c+1)}$.
Thus there exist $\sigma_1,\dots,\sigma_k$ such
that $|\sigma_i| \leq f(c+1) - c-1$ and
$M(\sigma_i) = \tau_i$ for each~$i\le k$.
Since there are only $2^{f(c+1) - c}-1$ strings
of length $\leq f(c+1)-c-1$, we have
$k < 2^{f(c+1) - c}$.
Since for each $i$, $\mu([\tau_i]) = 2^{-f(c+1)}$,
\[
\mu(U_c) ~=~ k \cdot 2^{-f(c+1)} ~<~ 2^{f(c+1)-c} \cdot 2^{-f(c+1)} = 2^{-c}.
\]
Therefore, $(U_c)_c$ is a primitive recursive test.
By assumption, $X \in U_c$ for all $c \geq 0$, so $X$
is not \ml $\BP$ random.

{(3)\,$\Rightarrow$\,(1)}: Suppose that $X$
is not \ml $\BP$ random. Then there exist primitive recursive functions
$g$, $k$, and $f$ so that for all $c \geq 0$, $g(c)$ is a code of a
set $G_c \subseteq \{0,1\}^{f(c)}$ with cardinality $k(c)$,
such that if $U_c = [G_c]$ then $\mu(U_c) \leq 2^{-c}$ and
such that $X \in \bigcap_c U_c$.  Furthermore, we may assume by
Proposition~\ref{prop2} that this is a weak test,
so that, for each $\sigma \in G_c$,
$\mu([\sigma] \cap U_{c+1}) \le \frac12 \mu([\sigma])$.
By the proof of Proposition~\ref{prop2}, we may assume
that $U_{c+1}\subset U_c$.
We may assume without loss of
generality that for each $c$, $f(c+1) - (c+1) > f(c) - c$.  This is
because we may always break each $[\tau]$ into $[\tau \fr 0] \cup
[\tau \fr 1]$ to increase $f(c)$ by one, if necessary.
Also w.l.o.g., $f(0)=0$, so $U_0 = \{0,1\}^\omega$.

We define a quick process $\BP$ machine $M$ such that
$C_M(\sigma) \leq f(c) - c$ for all~$c$ and all $\sigma\in G_c$.
Since $X\res f(c) \in G_c$, we have
$C_M(X \res f(c)) \leq f(c) - c$ for all~$c$.
The machine~$M$ is defined in stages as follows.

At stage $c=1$, we have $\mu(U_1) \leq \frac 12$, and since
$\mu(U_1) = k(1) \cdot 2^{-f(1)}$, it follows that
$k(1) \leq 2^{f(1)-1}$. Let $G_1 = \{\tau_1,\dots,\tau_{k(1)}\}$.
Take the lexicographically first $k$ strings
$\sigma_1, \dots, \sigma_{k(1)}$ of length $f(1) - 1$,
and define $M(\sigma_i) = \tau_i$.
To make $M$ a total function, the remaining strings of length
$f(1) - 1$ are all mapped to~$0^{f(1)}$, and all strings of length
$< f(1) - 1$ are mapped to the empty string~$\emptyset$.

Observe that for all strings~$\sigma$ of length $< f(1)-1$,
$|M(\sigma)| = 0$ and for all strings $\sigma$ of length $f(1) -1$,
$|M(\sigma)| = f(1)$.

After stage~$c$, we have defined $M(\sigma)$ for all strings~$\sigma$
of length $\leq f(c)-c$ so that $M$ is extension-preserving and such
that, for each $b \leq c$ and each $\tau \in G_b$, there exists
a $\sigma$ of length $f(b)-b$ with $M(\sigma) = \tau$. Furthermore, for
any~$\sigma$, if $f(b-1)-b+1 \leq |\sigma| < f(b)-b$, then
$|M(\sigma)| = f(b-1)$, and if $|\sigma| = f(c)-c$, then
$|M(\sigma)| = f(c)$.

At stage $c+1$, we define $M(\sigma)$ when
$f(c)-c < |\sigma| \leq f(c+1)-c-1$.  For each $\nu$ of length
$f(c)-c$ we work on the extensions of $\nu$ independently.
We set $M(\sigma) = M(\nu)$ for any extension~$\sigma$ of~$\nu$
of length $< f(c+1)-c-1$. If $M(\nu) \notin G_c$, then we let
$M(\sigma) = M(\nu) \fr 0^{f(c+1)-1-f(c)}$ for
all $\sigma\sqsupset\nu$ with $|\sigma| = f(c+1)-c-1$.
Otherwise, $M(\nu)\in G_c$ and let $H_\nu$ be the members of~$G_{c+1}$
which extend $M(\nu)$:
\[
H_\nu ~=~ \{\rho \in \{0,1\}^{f(c+1)-f(c)}: M(\nu) \fr \rho \in G_{c+1}\}.
\]
Since $(U_c)_c$ is a weak test,
$\mu(H_{\nu}) \leq \frac12$ and so $k_\nu = |H_\nu| \le 2^{f(c+1)-f(c)-1}$.
Enumerate $H_\nu$ in lexicographic order
as $\{\rho_1,\ldots,\rho_{k_\nu}\}$.
Let $\sigma_i$ be the $i$-th string of length~$f(c+1)-f(c)-1$
in lexicographical order with $i\le k_\nu$
and set $M(\nu\fr \sigma_i) =M(\nu)\fr \rho_i$.
For all other $\sigma\in \{0,1\}^{f(c+1)-f(c)-1}$, set
$M(\nu\fr \sigma) =M(\nu)\fr 0^{f(c+1)-f(c)}$.

It is clear that $M$ continues to be
extension-preserving.
Since $U_{c+1}\subset U_c$, each
$\tau\in G_{c+1}$ is equal to
$M(\nu) \fr \rho$ for some~$\nu$ and some $\rho \in H_{\nu}$. It
follows that $\tau = M(\sigma)$ for some~$\sigma$
of length $f(c+1)-c-1$.
Finally, for any~$\sigma$, if
$f(c)- c \leq |\sigma| < f(c+1)-c-1$, then $|M(\sigma)| = f(c)$ and
also if $|\sigma| = f(c+1)-c-1$, then $|M(\sigma)| = f(c+1)$.

To see that $M$ is primitive recursive,
observe that,
since $f(c+1) -\penalty10000 c -\penalty10000 1 > f(c) - c$ for all~$c$,
we have $f(c) - c \ge c$. Thus, given a string~$\sigma$ of length~$m$,
we need check only values $c \le m$
to find the least~$c$ such that $m \leq f(c) - c$.
Then we simply run the process above for $c$~stages to compute $M(\sigma)$.

To verify that $M$ is a quick process machine, we define a function~$h(m)$
so that $|M(\tau)| \geq h(|\tau|)$ for all $\tau$:  First
let $h'(m)$ be the least $c \leq m$ such that $m < f(c+1) - c -1$ and
then let $h(m) = f(h'(m))$.

By assumption, $X \in U_c$ for every~$c$, so $X \res f(c) = \tau$
for some $\tau \in G_c$ hence $M(\sigma) = \tau$ where
$|\sigma| = f(c) - c$.
It follows that $C_M(X \res f(c)) = f(c) - c$.
Hence, $X$ is not quick process $\BP$ random.
\end{proof}

The proof of (3)\,$\Rightarrow$\,(1) above showed that the
function~$f$ giving the length of strings for
the \ml tests can be strictly increasing without loss
of generality.  The same function~$f$ was then
used for showing that $X$ can be primitive recursively compressed.
It follows that in the definition of Kolmogorov $\BP$ random, $f$~may
be assumed to be strictly increasing without loss of generality.

\begin{theorem} \label{bps2}
The following statements are equivalent for $X \in \{0,1\}^\omega$.
\begin{enumerate}
\item $X$ is \ml $\BP$ random.
\item $X$ is martingale $\BP$ random.
\end{enumerate}
\end{theorem}

\begin{proof}
{(1)\,$\Rightarrow$\,(2)}: Suppose that $X$ is not martingale $\BP$ random. Then
there is a primitive recursive martingale~$d$
and a primitive recursive function~$f$ such that,
for all~$n$, $d(X \res\penalty10000 f(n)) \geq 2^n$.
Let $G_n = \{\tau \in \{0,1\}^{f(n)}: d(\tau) \geq 2^n\}$ and $U_n = [G_n]$.
The sequence $(U_n)_{n \geq 0}$ is a primitive recursive sequence
of clopen sets, and $X \in \bigcap_n U_n$.

Since $d$ is a martingale and $d(\emptyset) = 1$,
$\sum_{|\tau| = m} d(\tau) \leq 2^m$.
It follows that there are at most $2^{f(n) - n}$ strings
$\tau \in \{0,1\}^{f(n)}$ such that $d(\tau) \geq 2^n$.
For each such~$\tau$, $\mu([\tau]) = 2^{-f(n)}$.
Thus $\mu(U_n) \leq 2^{f(n) - n} \cdot 2^{-f(n)} = 2^{-n}$.
Hence $(U_n)_{n \geq 0}$ is a primitive recursive test which
succeeds on~$X$, and $X$ is not \ml $\BP$  random.

{(2)\,$\Rightarrow$\,(1)}: Suppose $X$ is not $\BP$ \ml random. By
Proposition~\ref{prop2}, $X \in
\bigcap_n U_n$ for some weak primitive recursive
test $(U_n)_{n \geq 0}$.  As usual, $U_n =[G_n]$, where
$G_n = \{\tau_{1,n},\ldots,\tau_{k(n),n}\}$
is primitive recursively computable.
Let $f(n)$ be the length of the strings $\tau_{i,n}$.

We define a martingale $d$ as follows.
For $n=1$, and given $U_1 = [\tau_{1,1}] \cup \dots \cup [\tau_{k,1}]$, we let
$d(\tau_{i,1}) = \frac{2^{f(1)}}{k}$ for $i = 1,\ldots,k$.  If $\tau
\in \{0,1\}^{f(1)}\setminus\{\tau_{1,1}, \ldots, \tau_{k,1}\}$,
then we let $d(\tau) = 0$.  Since $\mu(U_1) \leq
\frac12$, it follows that $k \leq 2^{f(1)-1}$ and therefore $d(\tau_{i,1})
\geq 2$ for each $i$. Moreover,
$\sum_{ \tau \in \{0,1\}^{f(1)}} d(\tau) = k \cdot \frac{2^{f(1)}}k = 2^{f(1)}$.
Now work backwards using the martingale equation
$d(\sigma) = \frac12 (d(\sigma\fr 0) + d(\sigma \fr 1))$
to define $d(\sigma)$ for all~$\sigma$ of length $\leq f(1)$.
It follows by induction that for all
$j \leq f(1)$, $\sum_{ \tau \in \{0,1\}^{j}}d(\tau)  = 2^j$ so that,
in particular, $d(\emptyset) = 1$.

Now suppose that $n\ge 1$ and we have defined $d(\tau)$ for all~$\tau$ with
$|\tau| \leq f(n)$ so that  $d(\tau) \geq 2^n$ for all $\tau \in G_n$. We
need to extend $d$ to strings of length $\leq f(n+1)$. For
$\sigma$ of length $f(n)$, we will define $d(\sigma\tau)$,
where $\sg\tau = \sg \fr \tau$, for all~$\tau$ of length $f(n+1) - f(n)$.
If $d(\sigma) = 0$, then
we simply let $d(\sigma\tau) = 0$ for all~$\tau$. Now  fix $\sigma \in
G_n$ with $d(\sigma) \geq 2^n$ and consider $H = \{\tau: \sigma\tau \in G_{n+1}\}$.
Since $(U_n)_n$ is a weak test,
$\mu([H]) \leq \frac12$.
Thus we may proceed as in the first case where $n = 1$ to
define a martingale $m$ such that $m(\emptyset) = 1$ and $m(\tau) \geq 2$
for all $\tau \in H$. Now extend the definition of~$d$ to the strings
extending~$\sigma$ by defining $d(\sigma\tau) = d(\sigma)  \cdot m(\tau)$.
Since $d(\sigma) \geq 2^n$ and, for $\tau \in H$, $m(\tau) \geq 2$, it
follows that for $\sigma\tau \in G_{n+1}$, $d(\sigma\tau)
\geq 2^{n+1}$.  It is easy to see that this extension obeys the
martingale equality, since, for any~$\tau$,
\[
d(\sigma\tau) = d(\sigma) \cdot m(\tau) = d(\sigma) \cdot \frac12
(m(\tau^\frown 0) + m(\tau^\frown 1)) = \frac12 \cdot (d(\sigma\tau^\frown 0)
+ d(\sigma\tau^\frown 1)).
\]
Since $X \in \bigcap_n U_n$, it follows that $d(X \res f(n)) \geq 2^n$
for each $n$ and hence $d$ succeeds primitive recursively on $X$.

It is easy to see that this defines a
primitive recursive procedure to compute $d(\sigma)$.  The first step
is to compute $f(n)$ for $n \leq |\sigma|$ until we find $n$ so that
$|\sigma| \leq f(n)$. Then we consider all extensions $\tau$ of $\sigma$
of length $f(n)$. We
 follow the procedure outlined above to compute $d(\sigma
\res f(i))$ for $i \leq n$, and, hence, compute $d(\tau)$
for all extensions $\tau$ of $\sigma$ of length $f(n)$. Finally we
backtrack using the martingale inequality to compute $d(\sigma)$ from
the values of such
$d(\tau)$. Thus $d$ is a primitive recursive martingale
and $X$~is not martingale $\BP$ random.
\end{proof}

The above proofs are
similar to those that have already appeared in the literature.
For example Downey, Griffiths and
Reid \cite{DGR} proved the equivalence
of analogues of (1) and~(3) of Theorem~\ref{bps}
in the setting of Kurtz randomness, and their proof could be modified
to prove the equivalence of (1) and~(3) in Theorem~\ref{bps}.
Similarly, Bienvenu and Merkle \cite{BM07} gave a proof the equivalence of
parts (1) and~(2) in the setting of Kurtz randomness and
their proof can be modified to work in our setting.

Given Theorems \ref{bps} and~\ref{bps2},
we define $X \in \{0,1\}^\omega$ to be
{\em $\BP$ random} if and only if $X$ is \ml $\BP$ random.
Since every $\BP$ test is also a Kurtz test and a computable test, it
follows that all Kurtz random and all computably random reals
are $\BP$ random.

It is clear that no primitive recursive set can be $\BP$ random. It was
shown by Jockusch \cite{Kautz} that Kurtz random sets are
\emph{immune}, that is, they do not include any c.e.\  subsets. Here is a
version of that result for $\BP$ randomness.

\begin{proposition} \label{prop4}
If $X$ is $\BP$ random, and $f$ is an increasing
primitive recursive function, then  $X$~does not contain the range of~$f$.
\end{proposition}

\begin{proof}
Suppose for the contrapositive that $X$
contains the range of~$f$.
Let $G_n = \{\sigma \in \{0,1\}^{f(n)}: (\forall i < n) (\sigma(f(i)) = 1)\}$,
and $U_n = [G_n]$. It is clear that $\mu([U_n]) = 2^{-n}$
so that $(U_n)_{n \geq 0}$ is a primitive recursive test. But then $X$
belongs to each~$U_n$ so $X$ is not $\BP$ random.
\end{proof}

\begin{theorem} \label{thm2}
There is a recursive real which is
$\BP$ random.
\end{theorem}

\begin{proof}
Let $(M_e,f_e)$ enumerate all pairs of primitive recursive functions, where
$M_e: \{0,1\}^* \to \{0,1\}^*$
and  $f_e: \N \to \N$ is strictly increasing.
Let $C_e$ denote $C_{M_e}$.
We define a recursive real~$X$ such that, for any~$e$,
there is some~$c$ such that $C_e(X\res f_e(c)) > f_e(c) - c - 1$.
Thus no primitive recursive machine can primitive recursively compress~$X$,
so $X$~is $\BP$ random.

The definition of $X$ is in stages, with $X$ defined as the union of a sequence
$\emptyset = \tau_0 \subseteq \tau_1 \subseteq \cdots$.
Let $n_k = |\tau_k|$.
We set $\tau_0 = \emptyset$ and $n_0 = 0$.

At stage $k+1$, we let $c = n_k$, and $n_{k+1} = f_k(n_k)$.
We are looking for an extension $\tau = \tau_{k+1}$ of $\tau_k$
of length $f_k(c)$ which does not
equal $M_k(\sigma)$ for any $\sigma$ of length $\leq f_k(c) - c - 1$.
There are $2^{f_k(c) - c}$ different extensions of~$\tau_k$ of length $f_k(c)$,
but only $2^{f_k(c) - c} -1$ strings of length $\leq f_k(c) - c - 1$.
Hence such a string~$\tau$ exists and we may compute it primitive
recursively (and in a space efficient manner) by the following
algorithm.
Enumerate all possible values for $\tau\in\{0,1\}^{n_{k+1}}$
starting with $0^{n_{k+1}}$.  For each value~$\tau$, compute $M(\sigma)$
for all strings~$\sigma$ of length $\leq f_k(c) - c -1$.
If the values $M(\sigma)$ are all distinct from the candidate
value for~$\tau$, then output that~$\tau$.
\end{proof}

Next we show that $\BP$ random reals satisfy the following analogue of
Ville's Theorem.

\begin{theorem} \label{thm3}
Let $X \in \{0,1\}^\omega$ be $\BP$ random and let $g$ be a primitive
recursive increasing function.  Then the sequence
$(X(g(0)), X(g(1)), X(g(2)), \ldots )$
is also $\BP$ random.
\end{theorem}

\begin{proof} Let $Y(n) = X(g(n))$, and suppose
that $Y$ is not $\BP$ random.
Let $(U_n)_{n \geq 0}$
be a primitive
recursive test such that $Y \in \bigcap_n U_n$.
As usual, we may assume that $U_n= [G_n]$ where
$G_n$ is a subset of $\{0,1\} ^{f(n)}$ for some
primitive recursive function~$f$.  In other words,
each $\tau \in G_n$ has length~$f(n)$.

We define a primitive recursive test $(V_n)_n$ for~$X$
by letting $V_n\subset [\{0,1\}^{g(f(n))}]$ be
\[
V_n ~=~ \{ X: (\,X(g(0)),X(g(1)),\dots,X(g(f(n)))\,) \in U_n\}.
\]
It is easy to see that $\mu(U_n)= \mu(V_n)$,
and $(V_n)_{n \geq 1}$ is a primitive recursive test.
Also, $X \in \bigcap_{n \geq 1} V_n$ which violates the assumption
that $X$ is $\BP$ random. This contradicts the
assumption that $Y$ is $\BP$ random.
\end{proof}

\subsection{Prefix-free primitive recursive}\label{sec:prefixfreeBP}

Kolmogorov complexity and randomness are usually studied for prefix-free machines.
By convention, primitive recursive functions are total, and thus not prefix-free.
Nonetheless, we can consider partial primitive recursive functions
by adding a new symbol,~$\infty$, for divergence.
The usual definitions of primitive recursive functions can be readily modified
to allow the primitive function to output the special symbol~$\infty$
to indicate the function diverges.
With this convention, a primitive recursive function~$M$ is {\em prefix-free}
provided that there do not exist distinct strings $\sigma\sqsubset\tau$
such that both $M(\sigma)\not=\infty$ and $M(\tau)\not=\infty$.

\begin{definition} A sequence $X\in\{0,1\}^\omega$ is
\emph{prefix-free $\BP$ random} if there does not exist
a prefix-free primitive recursive
function~$M$
and a primitive recursive function~$f$ such that
$C_M(X \res f(c)) \leq f(c) - c$ for all~$c$.
\end{definition}

\begin{proposition} \label{p2}
A real $X$ is $\BP$ random if and only if it is prefix-free $\BP$ random.
\end{proposition}

\begin{proof}
Certainly if $X$ is $\BP$ random then it is prefix-free $\BP$ random.
So suppose that $X$ is not $\BP$ random.
By Theorem~\ref{bps}, there is a primitive recursive test
$\{U_c: c \in \N\}$ such that $X \in \bigcap_c U_c$.
We may assume (by  replacing $U_c$ with $U_{2c}$ if necessary)
that in fact $\mu(U_c) \leq 2^{-2c}$.
As usual, $U_n$~is defined in terms of a
primitive recursive $G_n = \{\tau_{1,n},\ldots,\tau_{k(n),n}\}$
where each $\tau_{i,n}$ has length $f(n)$, so
$U_n = [G_n]$.
As in the proof of Theorem~\ref{bps}, we
may assume that $f(c+1) - (c+1) > f(c) - c$.
From $|U_c|<2^{-2c}$, we have $k(c)2^{-f(c)}\le 2^{-2c}$.
Hence $k(c)\le 2^{f(c)-2c}$.
It follows that $f(c)-c\ge c$ for all~$c$.

We need to define a prefix-free primitive recursive function~$M$
such that $C_M(X \res f(c)) \leq f(c) - c$ for all~$c$.
We start by defining $M(\sigma)$ for $\sigma$ of length $< f(1)$.
Let $\sigma_{1,1}, \sigma_{2,1}, \dots, \sigma_{k(1),1}$ be the
lexicographically first $k(1)$ strings  of length $f(1) - 1$,
and define $M(\sigma_i) = \tau_{i,1}$.
And, let $M(\sigma)=\infty$
for all other strings of length $< f(1)$.
Note that $M(\sigma) \not= \infty$ for at most half
of the strings of length $f(1)-1$, since $k(1)\le 2^{f(1)-2}$.

Now suppose that we have defined the partial function~$M$,
in a prefix-free way, for strings of length $\leq f(c-1)-c+1$, so that
the following hold: First, $M$ is injective (where it converges).
Second, for each $n<c$
and $\tau_{i,n}\in G_n$, there is a unique $\sigma$ of length $f(n)$
such that $M(\sigma) = \tau_{i,n}$.  Third, $M(\tau)$ is undefined for
all other strings $\sigma$ of length $\le f(c-1)-c+1$.
Fourth, for each~$n<c$, $M(\sigma) \not= \infty$ for at most fraction $2^{-n}$
of the strings of length $f(n)-n$.
We want to extend $M$ by setting the values of $M(\sigma)$ for
all $\sigma$ of length $<f(c)-c$.  For $f(c-1)-c+1 < |\sigma| < f(c)-c$,
set $M(\sigma)=\infty$.  In order to define
$M(\sigma)$ for $|\sigma|=f(c)-c$, enumerate the
lexicographically first strings
$\sigma_{1,c},\sigma_{2,c},\ldots,\sigma_{k(c),c}$
of length $f(c)-c$ which do not extend any $\nu$ for which
$|\nu|<f(c)-c$ and
$M(\nu) \not=\infty$.
Once these are enumerated, define $M(\sigma_{i,c}) = \tau_{i,c}$.
To see that it is possible to enumerate such strings, first note that $k(c)\le 2^{-c}\cdot\penalty10000 2^{f(c)-c}$.
Furthermore, for each $n<c$, at most a fraction
$2^{-n}$ of the strings~$\nu$ in $\{0,1\}^{f(n)-n}$
have $M(\nu)\not= \infty$. Since $\sum_{n=0}^{c-1}  2^{-n}<1 -2^{-c+1}$,
there are more than $k(c)$ strings. In fact,
there are at least $2 k(c)$
available to enumerate.  Thus $M$ is well-defined.

By assumption, $X \in U_c$ for every~$c$,
so $X \res f(c) = \tau_{i,c}$ for some $i$ and hence
$M(\sigma_{i,c}) = \tau_{i,c} = X \res f(c)$.
Since $|\sigma_{i,c}| = f(c) - c$, we have
$C_M(X \res f(c)) = f(c) - c$ as desired.

It remains to check that $M$ is indeed a primitive recursive function.
Recall that $f(c) -\penalty10000 c \ge c$. Thus, given a string~$\sigma$,
we only need to check $c \le |\sigma|$ to see
whether $m = f(c) - c$ for some $c$.  If not, then $M(\sigma)=\infty$.
If so, we use the above construction of $M$ to determine whether
$\sigma = \sigma_{i,c}$ where $M(\sigma_{i,c}) = \tau_{i,c}$.
This is clearly a primitive recursive procedure for computing $M(\sigma)$.
\end{proof}

\subsection{Statistical tests}\label{sec:statisticaltests}

It is interesting to see to what extent the $\BP$ random sets are statistically random.
We begin with a positive result.

\begin{theorem} \label{thm4}
Let $X$ be a $\BP$ random set. For any increasing
primitive recursive function~$f$ and any $\epsilon > 0$,
\[
\left| \frac{card(X \cap [[f(n)]])}{f(n)} - \frac12 \right| \leq \epsilon
\]
for infinitely many~$n$.
\end{theorem}

\begin{proof}
This follows from the law of large numbers (Chernoff's
Lemma~\cite[p.\,61]{LV08}). For a finite string~$\sigma$,
let $card(\sigma)$ denote
$card(\{i: \sigma(i) = 1\})$.
For $n \in \N$ and any real $\epsilon > 0$, let
\[
S_{n,\epsilon} ~=~
   \Bigl\{ \sigma \in \{0,1\}^n:
     \Bigl| \frac{card(\sigma)}n - \frac 12\Bigr| >\epsilon \Bigr\}.
\]
Chernoff's Lemma states that $\mu([S_{n,\epsilon}])  \leq 2e^{-\epsilon^2 n/6}$.

Fix $\epsilon>0$ and let $m = \lceil 1/\epsilon\rceil$.
Let $V_n = S_{6 m^2 (n+1)\ln 2,\frac1m}$. It follows by an easy
calculation from Chernoff's Lemma that $\mu(V_n) \leq 2^{-n}$.
Finally, let $U_n = V_{f(n)}$.
Since $f$ is increasing, $f(n) > n$ and hence
$\mu(U_n) \leq 2^{-f(n)} \le 2^{-n}$ for all $n$.
We claim that $(U_n)_n$ is in fact a primitive recursive test.
This is because $U_n = [G_n]$, where $G_n$ is a set
of strings of length $6 m^2 f(n)$ and the membership of $\tau$ in~$G_n$
is easily computable by counting the numbers of 0's and 1's in~$\tau$.

Since $X$ is $\BP$ random, it follows that $X \notin U_n$ for at least one~$n$.
In fact, by considering the tests $(U_{n+i})_n$
for each $i \in \N$, we see that $A \notin U_n$ for infinitely many~$n$.
\end{proof}

\begin{corollary} \label{cor1} For any $\BP$
random set~$X$, if $\lim_n card(X \cap [[n]])/n$ exists, then it
equals $1/2$.
\end{corollary}

On the other hand, $\BP$ random sets do not have to be stochastic.
Note that this is also the case for Kurtz random sets.

\begin{theorem} \label{thm5}
There is a computable $\BP$ random set~$X$ such that
\hbox{$\lim_n card(X \cap [[n]])/n$} does not exist.
\end{theorem}

\begin{proof} To construct such a set~$X$, modify the proof of
Theorem~\ref{thm2} by adding long strings of~0's and long strings of~1's (in
alternation) after satisfying each requirement. Then we can make the
density go arbitrarily low and then arbitrarily high infinitely
often.
\end{proof}

\subsection{Relative randomness}\label{sec:relativeBP}

We now turn to the concept of {\em relative} primitive recursive randomness,
namely, $\BP$ randomness relative to an oracle~$Y$.  Let
$Y = (Y(0),Y(1), Y(2),\ldots )$ be a real.  The set of
functions which are primitive recursive relative to~$Y$ can be
defined in two equivalent ways.  The first definition
uses the characterization of the primitive
recursive functions as the closure of a set of base functions
under composition and primitive recursions.  Namely,
let the {\em base} primitive recursive functions
consist of the constant functions $f(\vec x) = c$ for $c\in \N$,
the successor function $f(x)= x+1$, and the projection
functions $f(x_1,\ldots, x_k) = x_i$.  Then a function
is {\em primitive recursive relative to~$Y$}  it can be obtained
from the base functions plus the function~$Y$ under
the usual primitive recursion operations of composition and
primitive recursion.

The second definition is in terms of Turing machines
with runtimes bounded by primitive recursive functions.  Namely,
a Turing machine, denoted~$M^Z$, is given oracle access
to the values of~$Z$ in the usual way
by giving it a query state and a oracle query/answer tape.
An oracle Turing machine~$M^Z$ is called {\em a primitive
recursive oracle machine} provided there is a
primitive recursive function~$g$, such that for
all~$x$ and all reals~$Z$, $M^Z(x)$~has runtime bounded by~$g(x)$
A function~$f^Z$ is a {\em primitive recursive oracle function}
provided there is some primitive recursive oracle machine~$M^Z$ such that
that for all~$x$ and all reals~$Z$, $f^Z(x) = M^Z(x)$.
Then, for a fixed oracle~$Y$,
a function~$f(x)$ is {\em primitive recursive relative to~$Y$}
if and only if there is some primitive recursive oracle Turing machine~$M^Z$
such that $f(x) = M^Y(x)$.

The equivalence of the two definitions for primitive recursive relative to~$Y$
is well-known, and depends on the fact that $Y$ is a 0/1-valued function
and thus is majorized by a primitive recursive function.

We shall define a real  $X =(X(0),X(1), \ldots, )$ to be
\ml $\BP$ random relative to~$Y$, Kolmogorov $\BP$ random relative
to~$Y$, and martingale $\BP$ random relative to~$Y$ by replacing
(some of) the primitive recursive functions in the definitions
of \ml $\BP$ random, Kolmogorov $\BP$ random, and martingale $\BP$ random
by primitive recursive functions relative to~$Y$, respectively.

\subsubsection{The compressibility definition}
Let $C_M^Y(\tau)$ be the length~$|\sigma|$ of the shortest
string~$\sigma$ such that
$M^Y(\sigma) = \tau$, that is, the length
of the shortest $M^Y$-description of~$\tau$.

\begin{definition}
An infinite sequence~$X$ is {\em primitive recursively compressed relative
to~$Y$} if there exist a primitive recursive oracle machine $M$ and a
primitive recursive function $f$ such
that, for every $c$, $C_M^Y(X\res f(c)) \leq f(c)-c$.
An infinite sequence~$X$ is \emph{Kolmogorov bounded primitive
recursively random relative to $Y$} (for short, {\em Kolmogorov $\BP$ random
relative to~$Y$}) if it cannot be primitive
recursively compressed relative to~$Y$.
\end{definition}

If in the above definition, the condition that $f$ is primitive recursive is replaced by $f$ is primitive
recursive relative to $Y$, then the set of  reals which are Kolmogorov bounded
primitively random relative to $Y$ would not change.
This is because we can always find
a primitive recursive function~$f^\prime$ (not relative to~$Y$)
that dominates~$f$,
and then use the construction mentioned immediately after
the proof of Theorem~\ref{bps}.

\subsubsection{The measure definition}
A {\em primitive recursive oracle test} is given by
a function~$g^Z$ which is primitive recursive relative to~$Z$,
and a primitive recursive function~$f$
such that, for all reals~$Z$,
(1)~$g^Z(n)$ codes a finite set, denoted $G^Z_n$ of strings,
with $G_n^Z \subseteq \{0,1\}^{f(n)}$,
(2)~$U_n^Z = [G_n^Z]$ is a clopen set, and
(3)~$\mu(U_n^Z) <2^{-n}$.  If these conditions hold,
then, for a particular real~$Y$,
we also refer to the sequences $(G_n^Y)_n$ and
$(U_n^Y)_n$ as being primitive recursive oracle tests.

\begin{definition}
A real~$X$ is {\em \ml $\BP$ random relative to a real~$Y$} if, for any primitive
recursive oracle test $(U_n^Y)_n$,
$X \notin U_n^Y$ for some~$n$.
\end{definition}
As before, the definition of \ml $\BP$ random relative to a real~$Y$ would be
unchanged if the function~$f$ were allowed to be primitive recursive
relative to~$Y$ instead of just primitive recursive.

As a simple example, suppose that $X$ is itself primitive recursive and let
$G_n^Z = [X \res n]$ for all $n$ and~$Z$, so that $f(n) = n$.
Then
$(G_n^Z)_{n \geq 0}$ is a primitive recursive oracle test,
and $X \in   [G_n^Z]$ for all $n$.
Thus $X$ is not \ml  $\BP$ random relative to \emph{any}~$Z$.

On the other hand, suppose that $X$ is \ml $\BP$ random and let $Y$ be
primitive recursive. We claim then that $X$ is \ml $\BP$ random relative to~$Y$.
If not, let $(U_n^Y)_n$ be a primitive recursive oracle test such that
$X \in U_n^Y$ for all~$n$.  Then in fact $(U_n^Y)_n$ is a
primitive recursive test since $Y$ is primitive recursive.
And $X$
fails this test, contradicting the assumption that $X$ is \ml $\BP$ random.

The notion of a \emph{weak} test may also be relativized to say that,
as before, for each $\tau \in G_n^Z$, $\mu(U^Z_{n+1} \cap [\tau])
\leq \frac12 \mu([\tau])$. The proof of Proposition~\ref{prop1}
relativizes so that we have the following proposition.

\begin{proposition}  \label{p1r}  $X$ is \ml $\BP$ random relative to~$Y$
if and only if it passes every weak
primitive recursive oracle test relative to the oracle~$Y$.
\qed
\end{proposition}

\subsubsection{The martingale definition}

A primitive recursive oracle function~$d^Z$ is a {\em primitive recursive oracle
martingale} if $d^Z$ is a
martingale for all reals~$Z$.
\begin{definition}
A real $X$ is \emph{martingale $\BP$ random relative to~$Y$}
if there
is no primitive recursive oracle martingale~$d^Y$ which
succeeds primitive recursively on~$X$.
\end{definition}

It is straightforward to check that the proof of Theorem~\ref{bps}
relativizes to prove the following.

\begin{theorem} \label{relbps}
The following are equivalent for $X ,Y \in \{0,1\}^\omega$.
\begin{enumerate}
\item $X$ is Kolmogorov $\BP$ random relative to~$Y$.
\item $X$ is \ml $\BP$ random relative to~$Y$.
\item $X$ is martingale $\BP$ random relative to~$Y$.
\end{enumerate}
\end{theorem}
When these conditions hold,
we say that $X$ is {\em $\BP$ random relative} to~$Y$.

\subsubsection{van Lambalgen's theorem}
We now discuss van Lambalgen's theorem for relative $\BP$ randomness.
We first prove that one direction of van Lambalgen's theorem
holds in this setting, and then give a counterexample for
the other direction.
Recall that if $A,B \subseteq \N$, then
$A \oplus B = \{2x: x \in A\} \cup \{2x+1: x \in B\}$.

\begin{theorem} \label{thm:LambalgenDirectionOne}
If $A\oplus B$ is $\BP$ random,
then $A$~is $\BP$ random relative to~$B$ and
$B$~is $\BP$ random relative to~$A$.
\end{theorem}

\begin{proof}
Suppose $B$ is not $\BP$ random relative to~$A$,
and there is a primitive recursive oracle test $(U_n^A)_n$
such that $B \in \bigcap_n U_n^A$.
The test $(U_n^A)_n$ is given by a primitive recursive
oracle function~$g^Z$ so that, for all reals~$Z$, $g^Z(n)$~codes
a set $G^Z_n\subset\{0,1\}^*$ so that $U_n^Z = [G^Z_n]$.
The runtime
of $g^Z(n)$ is primitive recursively bounded; in particular,
there is a primitive recursive function~$\ell(n)$ so that,
for all~$Z$ and all~$n$, $g^Z(n)$~only queries values of~$Z(i)$
for $i<\ell(n)$. In addition, $G_n^Z\subset\{0,1\}^{\le \ell(n)}$
for all~$n$.
For $z=\{0,1\}^*$ with $|z|\ge \ell(n)$, we
can thus unambiguously define $g^z(n)$ to equal the value of $g^Z(n)$
when run on any $Z \sqsupset z$.  In this situation, we
write $G^z(n)$ for the set coded by~$g^z(n)$, so
$G^z_n = G^Z_n$ for any $Z\sqsupset z$.
We define a primitive recursive test $(V_n)_n$ by letting
$V_n = [H_n]$ where
\[
H(n) ~=~ \{ a\oplus b :
   \hbox{$a,b\in\{0,1\}^{\ell(n)}$ and $b \in [G^a_n]$}\}.
\]
Clearly, $H_n$ is primitive recursive since $b \in [G^a_n]$
holds iff some prefix of~$b$ is in $G^a_n$.  Furthermore,
since for each $a\in \{0,1\}^{\ell(n)}$, $\mu([G^a_n])\le 2^{-n}$,
we have
$\mu(V_n) \le 2^{-n}$.  In addition,
since $B\in U_n^A$ for all~$n$, we also have
$A\oplus B$ is in~$V_n$ for all~$n$.  Therefore $A\oplus B$
is not $\BP$ random.

It follows by symmetry that
if $A\oplus B$ is $\BP$ random, then
$A$~is $\BP$ random relative to~$B$.
\end{proof}

Next we give a counterexample to the converse of
Theorem~\ref{thm:LambalgenDirectionOne}.

\begin{theorem}\label{thm:Lambalgen:CounterExample}
There are reals $A$ and~$B$ such that $A$~is $\BP$ random relative to~$B$,
and $B$~is $\BP$ random relative to~$A$, but $A\oplus B$ is
not $\BP$ random.
\end{theorem}

\begin{proof}
We construct the reals $A$ and~$B$ in stages.
The stages will be controlled by a
fast growing increasing function~$h(i)$:
at the end of stage~$i$,
the values of $A(n)$ and~$B(n)$ for $n<h(i)$ will
have been set.  Furthermore, these values will
satisfy:
\begin{enumerate}[label=\arabic*.]
\item
For $h(2j)\le n < h(2j+1)$, we have $B(n)=0$.
\item
For $h(2j+1)\le n < h(2j+2)$, we have $A(n)=0$.
\end{enumerate}
Define the sequences $A_i,B_i\in\{0,1\}^{h(i)}$
to consist of the values of $A(n)$ and $B(n)$ that
have been set by the end of the $i$-th stage.  In the end,
we set $A = \lim_i A_i$ and $B = \lim_i B_i$.
Conditions 1\ and~2\  clearly imply that $A\oplus B$
is not $\BP$ random.

Let $(M^Z_j,f_j)_{j\ge 0}$ enumerate all pairs such that
$M^Z_e:\{0,1\}^*\mapsto\{0,1\}^*$
is a primitive recursive oracle function,
and $f_j:\N\rightarrow\N$ is a
primitive recursive function with $f_j(c)>c$ for all~$c$.
Stage~$2j+1$, which sets the values of $A(n)$
and~$B(n)$ for $h(2j)\le n < h(2j+1)$,
will ensure that $A$ is not primitive recursively
compressed by $(M^B_j,f_j)$.
Stage~$2j+2$
will act similarly
to ensure that $B$ is not primitive recursively
compressed by $(M^A_j,f_j)$.

We describe an odd stage~$2j+1$.  The value
of $h(2j)$ has already been set; we must define $h(2j+1)$
and the values of $A(n)$ for $h(2j)\le n < h(2j+1)$.
Define $X$ to be the infinite sequence consisting
of $B_{2j}$ followed by all~$0$'s.
Let $c=h(2j)+1$ and compute $f_j(c)>c$.
There are $2^{f_j(c)-h(2j)}$ many strings of the
form $A_{2j} u$ with $|A_{2j} u| = f_j(c)$.
On the other hand, there are at most
$2^{f_j(c)-h(2j)}-1$ strings of the form
$M^X_j(v)$ with $|v|\le f_j(c)-h(2j)-1 = f_j(c)-c $.
Fix a value for~$u$ so that $|A_{2j} u| = f_j(c)$
and $u$ is not equal to $M^X_j(v)$ for any
$v\in\{0,1\}^{\le f_j(c)-c}$.

Now, set $h(2j+1)$ to be equal to the
least value $\ge f_j(c)$ such that,
for every $v\in\{0,1\}^{\le f_j(c)-c}$,
$M^X_j(v)$ only queries oracle values
$X(m)$ for $m<h(2j+1)$.
This ensures that $M^X_j(v) = M^B_j(v)$
for all such~$v$'s.
Set $A_{2j+1} = A_{2j} u 0^{h(2j+1)-f_j(c)}$,
so $|A_{2j+1}| = h(2j+1)$.
Further set
$B_{2j+1} = B_{2j} 0^{h(2j+1)-h(2j)}$.
By construction,
$C_{M_j^B}( A\upharpoonright f_j(c) ) > f_j(c)-c$,
so $A$~is not
primitive recursively
compressed by $(M^B_j,f_j)$.

That completes the description of
the even stages.  The odd stages are defined similarly.
\end{proof}

Kihara and Miyabe \cite{KM13}  have obtained a similar result for Kurtz random reals,
using the same basic idea for the proof.

\section{Polynomial-space bounded randomness}\label{sec:BPSrandomness}

In this section, we modify the definitions of $\BP$ randomness to define a notion
of bounded randomness relative to polynomial space functions, called
 ``$\BPS$ randomness''.
It is important to note that this is a notion of bounded randomness,
similar to the bounded primitive recursive randomness developed above.

The usual notion of $\PSPACE$ randomness is part of the well-developed study
of resource-bounded measure and randomness due to
Lutz, Schnorr, and Ambos-Spies and Mayordomo, see \cite{Lutz90,AM97,S71}.
It will follow from our definitions that $\PSPACE$ randomness implies $\BPS$ randomness.
The two notions are certainly different, as $\PSPACE$ random sets satisfy the law of large numbers,
whereas $\BPS$ random sets, like $\BP$ random sets,
satisfy only a limited form of this law, see, for example, Theorem 3.13.

We begin by defining three equivalent notions of $\BPS$ randomness,
analogous to the three versions of $\BP$ randomness.

First, however, we need to define
polynomial space functions.
A $\PSPACE$ {\em predicate}, or {\em set}, is one for which membership can be
decided by a Turing machine which uses space bounded by a polynomial
$p(n)$ of the length~$n$ of its input.
A polynomial space computable {\em function}
is a function $\{0,1\}^* \rightarrow \{0,1\}^*$ which can be computed
by a Turing machine~$M$ with a read-only input tape, work tapes,
and a write-only output tape such that the space that $M$ uses on
its work tapes is polynomially bounded by the length~$n$ of  its input.
It is required that $M$ halts for all possible inputs, thus the
runtime of~$M(x)$ is bounded by $2^{p(n)}$ for some polynomial~$p$,
where $n=|x|$ is the length of the input~$x$.
If in addition, the length of output of~$M$ is polynomially
bounded, then we say $M$ computes a function in $\FPSPACE$.
That is, an $\FPSPACE$ function~$f$ is computable in polynomial
space, and there is a polynomial~$p(n)$
such that $|f(x)| \le p(|x|)$ for all inputs~$x$.
On the other hand, if we do not bound the length of $f(x)$,
so that $|f(x)|$ is bounded only by $2^{p(n)}$ for some
polynomial~$p$, then we call $f$ an $\FPSPACEplus$ function.
However, we are primarily interested in $\FPSPACE$ functions.

We will use $\FPSPACE$ functions for defining all three
of \ml, Kolmogorov, and martingale $\BPS$ randomness.
(So, $\FPSPACEplus$ functions
are not needed for the definitions.)
All three types of definitions use
$\PTIME$ functions $f:\{1\}^*\rightarrow \{1\}^*$
to
bound the lengths of strings.  The importance of these $\PTIME$ functions
lies solely in the fact that their growth rate is bounded by a polynomial;
in other words, in the fact that
$|f(1^n)| \le p(n)$ for some polynomial~$n$.  In fact, all our definitions
could equivalently use just
functions of the form $f(1^n) = 1^{n^c+c}$ for $c\in \N$.
Similarly, the definitions could equivalently use
$\FPSPACE$ functions $f:{1}^*\rightarrow \{0,1\}^*$, as their growth
rate is also bounded by a polynomial.

\begin{definition}
A $\PSPACE$ test $(U_n)_{n \geq 0}$ is specified by
a pair of functions $(G,f)$ such that
$G:\{1\}^* \times \{0,1\}^* \rightarrow \{0,1\}$
is an $\FPSPACE$-function
and $f:\{1\}^* \rightarrow \{1\}^*$ is a strictly length
increasing $\PTIME$ function
such that for each~$n$,
\[
G_n ~=~ \{\tau \in \{0,1\}^{\leq |f(1^n)|}: G(1^n,\tau)= 1\}
\]
is a set of strings of length
$\leq |f(1^n)|$ such that $U_n = [G_n]$ is a clopen set with measure $ \leq 2^{-n}$.
\end{definition}

As is easy to check,
$\PSPACE$ tests could be equivalently defined with a
$\FPSPACEplus$ function~$g$ instead of the $\FPSPACE$ function~$G$,
by defining $G_n$ to be equal to the set coded by the (potentially
exponentially long) string $g(1^n)$.

The above definition gave $G_n$ as the domain of
a parameterized $\FPSPACE$ function~$G$.  An alternate,
and equivalent, definition is to define~$G_n$
as the range of a parameterized $\PSPACE$ function~$\widehat G$:

\begin{proposition}\label{prop:PspaceRange}
The sequence $(G_n)_n$ and~$f$ satisfy the conditions for the
definition of a $\PSPACE$ test iff there is an
$\FPSPACE$ function
$\widehat G :\{1\}^* \times \{0,1\}^* \rightarrow \{0,1\}^*$
such that
\[
G_n ~=~
\{\sigma \in \{0,1\}^{\leq |f(1^n)|}:
    \text{for some}\ \tau\in \{0,1\}^{\leq |f(1^n)|},\,
     \widehat G(1^n,\tau)= \sigma\}.
\]
\end{proposition}
\noindent
Therefore, we will also refer to $(\widehat G, f)$ as
being a $\PSPACE$ test.

%

A \emph{weak} $\PSPACE$ test $(U_n)_{n \geq 0}$ is a $\PSPACE$ test
as above with the additional property that for each~$n$ and $\sigma_{i,n}\in G_n$,
$\mu(U_{n+1} \cap [\sigma_{i,n}]) \leq \frac12 \mu([\sigma_{i,n}])$.

\begin{definition}
An infinite sequence~$X$ is {\em \ml $\BPS$ random}
if $X$ passes every $\PSPACE$ test.  $X$~is
{\em weakly \ml $\BPS$ random} if $X$ passes every weak $\PSPACE$ test.
\end{definition}

\begin{definition}
An infinite sequence~$X$ is {\em Kolmogorov $\BPS$ random} if there do not
exist an $\FPSPACE$ function $M:\{0,1\}^* \rightarrow \{0,1\}^*$ and a
$\PTIME$ function $f:\{1\}^* \rightarrow \{1\}^*$ such that,
for every $n \in \N$, $C_M(X\res |f(1^n)|) \leq |f(1^n)|-n$.
\end{definition}

For the next definition, we take members of $\mathbb{Q} \cap [0,\infty)$
as being coded as ternary strings $\sigma 2 \tau$ where
$\sigma.\tau$ is the binary expansion
of a rational number in $\mathbb{Q} \cap [0,\infty)$.

\begin{definition}
An $\FPSPACE$ martingale
$d:\{0,1\} \rightarrow \mathbb{Q} \cap [0,\infty)$
{\em succeeds} on~$X$ if there is a $\PTIME$
function $f:\{1\}^* \rightarrow \{1\}^*$
such that, for all~$n$,
$d(X \res m) \geq 2^n$ where $m = |f(1^n)|$.
An infinite sequence~$X$ is
{\em martingale $\BPS$ random} if no $\FPSPACE$ martingale
succeeds on~$X$.
\end{definition}

We have the following analogue of Proposition~\ref{prop3}:
\begin{proposition}\label{prop:martingaleEquivsBPS}
The following are equivalent:
\begin{enumerate}
\item $X$ is martingale $\BPS$ random.
\item There do not exist an $\FPSPACE$
martingale~$d$ and a $\PTIME$ function~$f$ such
that, for every~$n$, there exists $m\leq f(n)$
such that $d(X \res m) \geq 2^n$.
\item There do not exist an $\FPSPACE$ martingale~$d$
and a $\PTIME$ function~$f$ such that $d(X \res m) \geq 2^n$ for
all $n$ and all $m \geq f(n)$.
\end{enumerate}
\end{proposition}
\begin{proof}
As before, the implications (2)\,$\Rightarrow$\,(1) and
(1)\,$\Rightarrow$\,(3) are immediate.  Proving
the implication (3)\,$\Rightarrow$\,(2) requires
a refined version of the argument used in the
proof of Proposition~\ref{prop3}.  As before,
we use the idea of a savings account, but now the
intuition is that once $d(\tau)\ge 4^{n+1}$, then
one half of the working capital
is transferred to a savings account.
So, loosely speaking, every time the capital
is {\em quadrupled}, one half of it is moved into
the savings account.
Suppose $d$ and~$f$ are as in~(2).  We formally define
a martingale~$d^\prime$ as follows.
Given $\tau\in \{0,1\}^*$, define
$\ell^\tau_n$ to equal the least value~$\le |\tau|$,
if any, such that $d(\tau\res \ell^\tau_n) \ge 4^{n+1}$.
Let $\tau_n$ equal $\tau \res \ell^\tau_n$, and let $n(\tau)$ be the
least value of~$n$ for which $\ell^\tau_n$ is undefined.
Note that $\ell^\tau_n$, $\tau_n$, and $n(\tau)$
are in $\FPSPACE$, since $d$ is in $\FPSPACE$.
Now define
\[
d^\prime(\tau) ~=~
   \frac{d(\tau)}{2^{n(\tau)}} +
   \sum\limits_{i=0}^{n(\tau)-1} \frac{d(\tau_i)}{2^{i+1}}.
\]
Once again, the summation
represents the capital placed in
the savings account, and
$d(\tau)/2^{n(\tau)}$ is the remaining working capital.
As before, it is easy to verify that $d^\prime$ is
a martingale.  Also, by the definition,
$d^\prime(\tau)$ is in $\FPSPACE$.

Since $d(\tau_i)\ge 4^{i+1}$,
we have $d(\tau)\ge \sum_{i=1}^{n(\tau)-1}2^{i+1} > 2^{n(\tau)}$.
Let $f^\prime$ be the $\PTIME$ function $f^\prime(1^n) = f(1^{2n})$.
Suppose $X\in\{0,1\}^\omega$ and $m\ge |f^\prime(1^n)|$.  Since
$m \ge |f(1^{2n})|$,
there is some $m^\prime\le m$ such that
$d(X \res m^\prime)\ge 2^{2n} = 4^n$.
Therefore, $n(X\res m)\ge n$, so $d^\prime(X \res m) \ge 2^n$.
This establishes
that the properties of~(3) holds for $d^\prime$ and~$f^\prime$.
\end{proof}

By suitably modifying the proofs of Theorem \ref{bps} and~\ref{bps2},
we can prove the equivalence of these three versions of
$\BPS$ randomness.  As we shall see, the modifications are straightforward
in the case of Theorem~\ref{bps}, but more substantial for
Theorem~\ref{bps2}.

As a preliminary step, the next lemma states that,
for \ml $\BPS$ randomness, the strings
in the sets~$G_n$ of a
$\PSPACE$ test can be required to all be the same length.

\begin{lemma} For any $\PSPACE$ test $(G,f)$,
there is an $\FPSPACE$ function~$G^\prime$ such that $(G^\prime, f)$
is a $\PSPACE$ test
such that, for all $n$, $G^\prime_n \subset \{0,1\}^{f(n)}$.
\end{lemma}

\begin{proof}
We are given that $G_n = \{ \sigma\in\{0,1\}^{\le f(n)} : G(1^n,\sigma)=1 \}$.
It will suffice to define an $\FPSPACE$ function $G^\prime$ so that
$G^\prime_n = \{ \sigma\in\{0,1\}^{f(n)} : G^\prime(1^n,\sigma)=1 \}$
satisfies $[G_n] = [G^\prime_n]$ for all~$n$.
This is simply done: define $G^\prime(1^n,\sigma) = 1$
to hold exactly when $|\sigma| = n$ and there is some $\tau\sqsubset \sigma$
such that $G(1^n,\tau) = 1$.  Clearly $[G_n] = [G^\prime_n]$,
and $G^\prime$~is in $\FPSPACE$ since $G$ is.
\end{proof}

\begin{definition}
A function $M: \{0,1\}^* \to \{0,1\}^*$
is said to be an {\em $\FPSPACE$ $\BPS$ quick process
machine} if $M$~is in $\FPSPACE$ and is a quick process machine
with order function~$h$ such that
there is a $\PTIME$ function $g: \{1\}^* \to \{1\}^*$ such
that $g(1^n) =1^{h(n)}$ for all~$n$.
An infinite sequence~$X$ is {\em quick process $\BPS$ random}
if there do not
exist an $\FPSPACE$ $\BPS$ quick process machine~$M$ and a
$\PTIME$ function $f:\{1\}^* \rightarrow \{1\}^*$ such that,
for every $n \in \N$, $C_M(X\res |f(1^n)|) \leq |f(1^n)|-n$.
\end{definition}

\begin{theorem} \label{thm:firstequivsBPS}
The following are equivalent for $X \in \{0,1\}^\omega$.
\begin{enumerate}
\item $X$ is \ml $\BPS$ random.
\item $X$ is Kolmogorov  $\BPS$ random.
\item $X$ is quick process $\BPS$ random.
\end{enumerate}
\end{theorem}

\begin{proof}
The proof is almost identical to the proof of
Theorem~\ref{bps}, but noting that all the
constructions preserve the property of
being computable in $\FPSPACE$, instead
of being primitive recursive.
As before, (2)\,$\Rightarrow$\,(3) is immediate.
The proof of (1)\,$\Rightarrow$\,(2) is identical
to the proof of the same case for Theorem~\ref{bps}.
The only difference is that now $G_c$ must be
the range of a parameterized $\FPSPACE$ function~$\widehat G$ with
$G_c = \{\sigma: \sigma = \widehat G(1^c,\tau),\ \tau\in\{0,1\}^{|f(1^n|)} \}$.
Referring back to the definition of~$G_c$ in the proof
Theorem~\ref{bps} this is immediately seen to hold, as $M$~is
in $\FPSPACE$.

The proof of (3)\,$\Rightarrow$\,(1) is also identical to the
corresponding proof of Theorem~\ref{bps}; however, it is required
to show that the function~$M$ is in $\FPSPACE$.  In fact, the construction
given in the proof Theorem~\ref{bps} gives a polynomial space
algorithm for~$M$.
Recall that the values of $M(\sigma)$
for $f(c)-c < |\sigma| \le f(c+1)-c-1$ were computed by
expressing $\sigma$ in the form $\sigma = M(\nu)\fr \rho$,
if possible, where
$f(c-1)-c+1 < |\nu| \le f(c)-c$, and then enumerating $H_\nu$
in lexicographic order.  Since $G_c$ is given as the parameterized
domain of the $\FPSPACE$ function~$G$, it is certainly possible
to enumerate $G_c$ and hence $H_\nu$ using only polynomials space.
It is also necessary to find the value~$\nu$, if any,
such that $M(\nu)\sqsubset \sigma$.  This is done by calculating
$M(\nu)$ for all $\nu\in \{0,1\}^{|f(c)|-c}$.

In other words, $M(\sigma)$ is computed in polynomial space
by a recursive procedure that needs to compute
values of $M(\nu)$ with $f(c-1)-c+1 < |\nu| \le f(c)-c$.
The depth of the recursion is equal to~$c$; hence, the overall
computation of~$M(\sigma)$ can be carried out in polynomial space.
\end{proof}

Our next goal is to show that
$X$ is martingale $\BPS$ random if and only if
$X$ is weakly \ml $\BPS$ random.
The following lemma about martingales is a
version of Kraft's Inequality and will be needed to help
us prove this result.  Two strings $\sigma,\tau\in\{0,1\}^*$
are said to be {\em incompatible} if they are distinct
and neither one is a prefix
of the other.

\begin{lemma} [Kraft's Inequality] \label{lem1}
For any martingale~$d$, any $\sigma \in \{0,1\}^*$ and any finite set~$H$
of pairwise incompatible extensions of~$\sigma$,
\[
\sum_{\tau \in H} d(\tau) \cdot 2^{-|\tau|}
   ~\leq~ d(\sigma) \cdot 2^{-|\sigma|}.
\]
\end{lemma}

\begin{proof} Let $n = \max\{|\tau|: \tau \in H\}$ and
note that by the martingale condition,
\[
\sum \{d(\rho): |\rho| = n\ \&\ \sigma \sqsubset \rho\}
   ~=~ d(\sigma) \cdot 2^{n - |\sigma|}.
\]
For each $\tau \in H$, let $G(\tau) = \{\rho \in \{0,1\}^n: \tau
\sqsubset \rho\}$. Then by the martingale condition as above
\[
\sum_{\rho \in G(\tau)} d(\rho) ~=~ d(\tau) \cdot 2^{n - |\tau|}.
\]
Thus
\begin{eqnarray*}
\sum_{\tau \in H} d(\tau) \cdot 2^{-|\tau|}
      &=& \sum_{\tau \in H} \sum_{\rho \in G(\tau)} d(\rho) \cdot 2^{-n} \\
      &\le& \sum\{ d(\rho) \cdot 2^{-n}: |\rho| = n\ \&\ \sigma \sqsubset \rho\}
         ~=~ d(\sigma) \cdot 2^{-|\sigma|}.
\end{eqnarray*}
\end{proof}

We next prove the analogue of Theorem~\ref{bps2} that
martingale $\BPS$ randomness is equivalent to the other
forms of $\BPS$ randomness.  The proof
of Theorem~\ref{bps2} needs considerable reworking
however.  This is primarily
because the proof of Theorem~\ref{bps2} depended on the
equivalence of weak \ml $\BPS$
randomness and \ml $\BPS$ randomness.  Unfortunately,
the proof of Proposition~\ref{prop2} cannot
be readily modified to apply to $\BPS$ randomness: the difficulty
is that that proof defined a function~$h$
so that $h(n+1) = \ell(h(n))+1$,
but $h$ may not be in $\PTIME$ even if $\ell$~is.
Nonetheless, it does follow from
the next theorem that any $\PSPACE$ test can be
converted into a weak $\PSPACE$ test.

\begin{theorem} \label{thm:martingaleEquivBPS}
The following are equivalent for $X \in \{0,1\}^\omega$.
\begin{enumerate}
\item $X$ is \ml $\BPS$ random.
\item $X$ is weakly \ml $\BPS$ random.
\item $X$ is martingale $\BPS$ random.
\end{enumerate}
\end{theorem}

\begin{proof} Clearly (1) implies (2).

{(2)\,$\Rightarrow$\,(3):} Suppose that (3) fails
 and $d$ is an $\FPSPACE$ martingale
which succeeds
on~$X$ with a $\PTIME$ function~$f$
such that, for all~$n$,
$d(X \res |f(1^n)|) \geq 2^n$.
%
Unlike the (2)\,$\Rightarrow$\,(1) case of Theorem~\ref{bps2},
we must construct a weak $\PSPACE$
test that $X$ fails instead of just a $\PSPACE$ test.
Define $G:\{1\}^* \times \{0,1\}^* \rightarrow \{0,1\}$ by
letting $G(1^n,\sg) =1$ iff
\[
|\sigma| \leq |f(1^n)|\ \&\ d(\sigma) \geq 2^n\ \&\ (\forall i<|\sigma|) (d(\sigma \res i) < 2^n).
\]
Note that, by the martingale equality, if $G(1^n,\sigma)=1$,
then $d(\sigma) < 2^{n+1}$.
Clearly $G$ is an $\FPSPACE$ function.  Thus
\[
G_n ~=~ \{\sigma: |\sigma| \leq |f(1^n)|\ \& \ d(\sigma) \geq 2^n\ \& \
(\forall i<|\sigma|)
(d(\sigma \res i) < 2^n)\}.
\]
Set
$U_n = [G_n]$; we have
$X \in U_n$ for all~$n$ by the assumption.
By Lemma~\ref{lem1} with $\sigma = \emptyset$,
$\sum_{\tau \in G_n} d(\tau) \cdot 2^{-|\tau|} \leq 1$. Since for
all $\tau \in G_n$, $d(\tau) \geq 2^n$, it follows that
 \[
\mu([U_n]) ~=~ \sum_{\tau \in G_n} 2^{-|\tau|} ~\leq~ 2^{-n}.
\]

Now let $V_n = U_{2n}$ for all $n \geq 0$  so that
$(V_n)_{n\geq 0}$ is a $\PSPACE$ test
that $X$ fails. We claim that
$(V_n)_{n \geq 0}$ is a weak
test.  We have $\mu(V_n) = \mu(U_{2n}) \leq 2^{-2n}$.
For $\sigma \in V_{n}$,
let $H(\sigma) = \{\tau: \sigma \sqsubset \tau\ \&\ \tau \in V_{n+1}\}$.
By Lemma~\ref{lem1},
\[
\sum_{\tau \in H(\sigma)} d(\tau) \cdot 2^{-|\tau|} ~\leq~ d(\sigma) \cdot 2^{-|\sigma|}.
\]
Since $d(\sigma) < 2^{2n+1}$ and for each $\tau \in H(\sigma)$,
$d(\tau) \geq 2^{2n+2}$, we obtain
\[
\sum_{\tau \in H(\sigma)} 2^{2n+2} \cdot 2^{-|\tau|}
   ~\leq~ 2^{2n+1} \cdot 2^{-|\sigma|}.
\]
Dividing by $2^{2n+2}$, we obtain
\[
\mu([H(\sigma)]) ~=~ \sum_{\tau \in H(\sigma)} 2^{-|\tau|}
   ~\leq~ \frac12 \cdot 2^{-|\sigma|}
   ~=~ \frac12 \mu([\sigma]).
\]
Thus the sequence $V_{n}$ is a weak $\PSPACE$ test as desired.
Thus if $X$ is not martingale $\BPS$ random, then $X$ is
not weakly \ml  $\BPS$ random.

{(3)\,$\Rightarrow$\,(1):}
This is the most difficult case to prove, since we have not (yet)
established the analogue of Proposition~\ref{prop2}
for $\BPS$ randomness.
Let $\{U_n\}_{n\ge 0}$ be a Martin-L\"of $\BPS$ test which $X$ fails,
with $U_n= [G_n]$ where $G_n$'s are as usual
given by an $\FPSPACE$ function~$G$.
In particular, each $U_n$ has measure $\mu(U_n)\le 2^{-n}$,
and w.l.o.g.\ every $\sigma\in G_n$ has length equal to $\ell(n)$ for
some $\PTIME$ function~$\ell$ such that $\ell(n)\ge i$.

We will define martingales $d_i$ for $i=1,2,\ldots$, and
then define the overall martingale which succeeds against~$X$
as $d = \sum_i 2^{-i} d_{i}$.
The martingale $d_i$ is defined using the set $U_{3i-1}$.  That is to say,
we use $U_j$ only for $j = 3i-1$.

\begin{definition}
First consider $i=1$.  The martingale $d_1$ will be defined from~$G_2$
so as to satisfy the following conditions:
\begin{enumerate}[label=\arabic*.]
\item $d_1(\emptyset) = 1$.
\item $d_1(\sigma) = 4$ for each $\sigma \in G_2$.
\item $d_1(\sigma 0) = d_1(\sigma 1) = d_1(\sigma)$ for
all $\sigma$ of length $|\sigma|\ge \ell(2)$.
\end{enumerate}
\end{definition}
The point of condition~3.\ is that the value of $d_1(\sigma)$
``settles down'' to a constant value which depends only
on the first $\ell(2)$ symbols of~$\sigma$.   Property~2.\ can
be forced to hold since $\mu(U_2) \le 1/4$, and hence the
martingale properties allow the value of $d_1(\emptyset)$ to
be quadrupled for $\sigma\in G_2$.

It is straightforward to define an $\FPSPACE$ function~$d_1$
satisfying these three conditions 1-3.  One way to do this
is to enumerate the members of~$G_2$: since $G_2\subset \{0,1\}^{\ell(2)}$
and $\mu([G_2]) \le 1/4$, we have $|G_2| \le 2^{\ell(2)-2}$.  Let $H_2$
be the lexicographically first $2^{\ell(2)-2} - |G_2|$ many strings of
length $\ell(2)$ which are not in~$G_2$.  Then define $d_1$ as
follows: For $\sigma = \ell(2)$, $d_1(\sigma)$
is set equal to~$4$ if $\sigma\in G_2 \cup H_2$,
and set equal to zero otherwise.
For longer strings~$\sigma$, $d_1(\sigma) = d_1(\sigma \res \ell(2))$.
For shorter strings~$\sigma$, $d_1(\sigma)$~is computed using
Kraft's inequality and counting the number of $\tau\sqsupset \sigma$
with $\tau\in G_2\cup H_2$.
\begin{definition}
Now consider $i>1$.  The martingale~$d_i$ is defined
so as to satisfy
\begin{enumerate}[label=\arabic*.i]
\item $d_i(\sigma) = 1$ for every $\sigma$
of length $|\sigma| < i$.  In particular, $d_i(\sigma)$ is
constant for small~$\sigma$'s.
\item $d_i(\sigma) = 2^{2i}$ for every $\sigma\in G_{3i-1}$.
\item $d_i(\sigma 0) = d_i(\sigma 1) = d_i(\sigma)$ for
all $\sigma$ of length $|\sigma|\ge \ell(3i-1)$.
\end{enumerate}
\end{definition}
The idea for making property~2.i\ hold with $i>1$ is similar to
the $i=1$ construction for~$d_1$.  Property~2.i\ is
possible since
\smallskip

\begin{enumerate}[label=\({\alph*}]
\item there are $2^{i-1}$ strings~$\sigma$
of length $i-1$,
and for these~$\sigma$'s, $d_i(\sigma) = 1$
and
\item we have $\mu(U_{3i-1}) \le 2^{-(3i-1)}$ so that
$|G_{3i-1}| \le 2^{\ell(3i-1)-(3i-1)}$.
\end{enumerate}

\noindent
For any~$\sigma$ of length~$i-1$,
\[
\mu(U_{3i-1} \cap [\sigma])/\mu([\sigma])
   ~\le~ \mu(U_{3i-1})/\mu([\sigma])
   ~\le~ 2^{-(3i-1)}/2^{-(i-1)} ~=~ 2^{-2i}.
\]
Thus $d_i$'s value of~$1$ for $\sigma$
may be multiplied by the factor $2^{2i}$
for each $\tau\in\penalty10000 U_{3i}$ which extends $\sigma$.
That is, define $H_{3i-1}$ to contain,
for each $\sigma\in\penalty10000 \{0,1\}^{i-1}$, the lexicographically
first $k(\sigma)$ many extensions of $\sigma$  in
\hbox{$([\sigma]\cap [\{0,1\}^{\ell(3i-1)}])\setminus G_{3i-1}$} where
$k(\sigma)$ equals $2^{2i} - |G_{3i-1}\cap [\sigma]|$.  Then define
$d_i(\sigma)$ as follows: For $\sigma\in \{0,1\}^{\ell(3i-1)}$,
set $d_i(\sigma)$ equal to $2^{2i}$ if $\sigma\in G_{3i-1}\cup H_{3i-1}$,
and equal to zero otherwise.  For shorter
strings~$\sigma$, set $d_i$ using Kraft's inequality; by construction,
this satisfies~1.
For longer strings~$\sigma$, $d_i(\sigma) = d_i(\sigma\res \ell(3i-1))$.
By construction, each $d_i(\sigma)$ can be computed in $\FPSPACE$,
uniformly in $i$ and~$\sigma$, by enumerating
the strings in~$G_{3i-1}$.

We finally claim that $d(\sigma)$ can
be computed in polynomial space.
For this,
\[
d(\sigma) ~=~ \sum_{i\ge 0} 2^{-i} d_i(\sigma) ~=~
  \sum_{0\le i\le |\sigma|} 2^{-i} d_i(\sigma) + \sum_{i > |\sigma|} 2^{-i} d_i(\sigma).
\]
The finite summation $\sum_{0\le i\le |\sigma|} d_i(\sigma)$ can be computed
in polynomial space by just evaluating each $d_i(\sigma)$.
The other, infinite, summation $\sum_{i>|\sigma|} d_i(\sigma)$ is
equal to~$2^{-|\sigma|}$ since $d_i(\sigma) = 1$
for $|\sigma|< i$.  Thus $d(\sigma)$ can be computed in
polynomial space.

Since $X{\upharpoonright}\ell(3i-1) \in U_{3i-1}$,
\[
d(X{\upharpoonright} \ell(3i-1)) ~\ge~ 2^{-i} d_i(X{\upharpoonright}\ell(3i-1))
  ~=~ 2^{-i} \cdot 2^{2i}
  ~=~ 2^i,
\]
and $d$ is an $\FPSPACE$ martingale which succeeds against~$X$.
This concludes the proof of Theorem~\ref{thm:martingaleEquivBPS}.
\end{proof}

We have the following analogue of Proposition~\ref{prop4}.
This is a version of bi-immunity for $\BPS$ random sets.
The notion of immunity is well-studied in the resource-bounded randomness
community. For example, it is known that any ptime random set is ptime bi-immune;
see Ambos-Spies and Mayordomo \cite{AM97}.

\begin{proposition} \label{prop:increasingRangeBPS}
Let $X$ be $\BPS$ random.  Suppose $h:\{1\}^*\rightarrow\{0,1\}^*$
is an
$\FPSPACE$ function
such that $f(n) = |h(1^n)|$ is an increasing function.
Then $X$~does not contain the range of~$f$.
\end{proposition}
\begin{proof}
Argue as in the proof of Proposition~\ref{prop4}.
\end{proof}

Next, analogously to Theorem \ref{thm2},
we construct a $\BPS$ random real which is not too complex.
A {\em $\DSPACE(t(n))$ real} is an infinite sequence $X\in\{0,1\}^\omega$
which is computable by a $\DSPACE(t(n))$ function (which is 0/1 valued, since
$X$ is 0/1 valued).

\begin{theorem} \label{thm2sp} Let $\epsilon>0$.  Then there is a
$\DSPACE(2^{\epsilon (\log n)^2})$  real which is
\ml $\BPS$ random.
\end{theorem}

\begin{proof} We modify the proof of Theorem \ref{thm2} as follows.
Let $(M_e,f_e)$ enumerate all pairs of functions such that
$M_e: \{0,1\}^* \to \{0,1\}^*$ is
computable in space $\max\{2,n\}^{\max\{2,\log e\}}$
and $f_e: \{1\}^* \to \{1\}^*$ is
computable in time $\max\{2,n\}^{\max\{2,\log e\}}$.
Every possible $\BPS$ compression will be found in this list, since
every computable function has infinitely many indices.
We may also assume without loss of generality that
$f_e(i) \geq i$ for all $e$ and $i$.
Let $C_e$ denote $C_{M_e}$.
We will define a real~$X$ such that, for every~$e$,
there is some~$c$ such that
$C_e(X\res f_e(c)) > f_e(c) - c - 1$.
It follows that no $\BPS$ machine can $\BPS$ compress~$X$,
and hence $X$ is $\BPS$ random.

The definition of~$X$ proceeds as in the proof of Theorem~\ref{thm2}. That is,
the definition of~$X$ is in stages, and $X$~is the union of a sequence
$\emptyset = \tau_0 \subseteq \tau_1 \subseteq \cdots$.
Let $n_k = |\tau_k|$.
Initially $\tau_0 = \emptyset$ and thus $n_0 = 0$.

At stage $k+1$, we let $c = n_k$, and $n_{k+1} = f_k(n_k)$.
Then we find an extension
$\tau = \tau_{k+1}$ of $\tau_k$ of length $f_k(c)$ which does not
equal $M_k(\sigma)$ for any $\sigma$ of length $\leq f_k(c) - c - 1$.
Observe that there are $2^{f_k(c) - c}$ different extensions
of~$\tau_k$ of length $f_k(c)$,
but there are only $2^{f_k(c) - c} -1$ strings of length $\leq f_k(c) - c - 1$.
Hence such a string~$\tau$ exists and we may compute it as follows.
Start with the guess that $\tau = 0^{n_{k+1}}$
and successively compute~$M_k(\sigma)$,
and compare it with~$\tau$, for every~$\sigma$ of length $\leq f_k(c) - c -1$.
If we find a match,
then change~$\tau$ to $0^{n_{k+1}-1} 1$ and
continue incrementing~$\tau$ until we find the desired string.

It remains to be seen that $X$ is in fact computable in space $n^{\epsilon \log n}$.
To compute $X \res n$ from~$n$, first compute the sequence
$$n_0 = 0, n_1 = f_0(n_0), \dots, n_{k+1} =\penalty10000 f_k(n_k)$$
until $n_{k+1} \geq n$.
Now $X \res n$ will be an initial segment of~$\tau_{k+1}$,
which will have length $n_{k+1} = f_k(n_k) \leq n^{\epsilon \log n}$,
for sufficiently large~$n$.
This last inequality is because $f_k(m)$ is computable in time $n^{\log k}$
and we have $n_k \leq n$ and also $k \leq n$ since $f_e(i) \geq i$.
Thus $M_k(\sigma)$ can be computed in space $\leq n^{\epsilon \log n}$
for all strings~$\sigma$ of length $\leq n_{k+1} - n_k - 1$.
It follows that the entire computation up through stage $k+1$,
needed to compute $X \res n$, can be done in space $n ^{\epsilon \log n}$.
\end{proof}

We also have the following analogue of Theorem~\ref{thm3}.
\begin{proposition}\label{thm:subsequenceBPS}
Let $X \in \{0,1\}^\omega$ be $\BPS$ random and
let $h:\{1\}^*\rightarrow\{0,1\}^*$ be
computable in $\FPSPACE$.  Suppose
that $g(n) = |h(1^n)|$ is an increasing function.
Then the sequence
$(X(g(0)), X(g(1)), X(g(2)), \ldots )$
is also $\BPS$ random.
\end{proposition}
\begin{proof}
The construction from the proof of Theorem~\ref{thm3} still
applies; we only need to verify that $V_n$ is an $\FPSPACE$
test.  To see this, just note that
$g(f(n))$ is computable in
space bounded by $p(n)$ for some polynomial~$p$.
From this, $V_n$ is a $\PSPACE$ test.
\end{proof}

We can also give a prefix-free characterization of
$\BPS$ randomness, completely analogous to the case of
$\BP$ randomness.
\begin{definition} A sequence $X\in\{0,1\}^\omega$ is
\emph{prefix-free $\BPS$ random} if there does not exist
a prefix-free $\FPSPACE$
function~$M$
and a $\PTIME$ function~$f$ such that
$C_M(X \res f(c) \leq f(c) - c$ for all~$c$.
\end{definition}

\begin{proposition} \label{prop:prefixfreeBPS}
A real~$X$ is $\BPS$ random if and only if it is prefix-free $\BPS$ random.
\end{proposition}
\begin{proof}
The proof of Theorem~\ref{p2} applies as is.  The only
change needed is to verify that the algorithm for~$M$ can
be carried out in polynomial space.  This is proved using the
same kind of argument used for the proof
of the (3)\,$\Rightarrow$\,(1) case of Theorem~\ref{thm:firstequivsBPS}.
\end{proof}

The theorems on statistical tests also carry over to the setting
of $\BPS$ random sets:
\begin{theorem} \label{thm:statisticsBPS}
Let $X$ be a $\BPS$ random set. For any increasing
primitive recursive function~$f$ and any $\epsilon > 0$,
\[
\left| \frac{card(X \cap [[f(n)]])}{f(n)} - \frac12 \right| \leq \epsilon
\]
for infinitely many~$n$.
\end{theorem}

\begin{proof}
The proof is identical to the proof of Theorem~\ref{thm4}.
The fact that $U_n$ is a $\PSPACE$ test follows from
the fact that membership in $S_{n,\epsilon}$ is decidable in
$\PSPACE$ for any fixed~$\epsilon$.
\end{proof}

\begin{corollary} \label{cor:statisticsBPS}
For any $\BPS$
random set~$X$, if $\lim_n card(A \cap [[n]])/n$ exists, then it
equals $1/2$.
\end{corollary}

\begin{theorem} \label{thm:nolimitBPS}
Let $\epsilon>0$.  Then there exists a
$\BPS$ random set~$A$ in
$\DSPACE(2^{\epsilon (\log n)^2})$ such that
$\lim_n card(A \cap [[n]])/n$ does not exist.
\end{theorem}

\begin{proof} To construct such a set~$A$, modify the proof of
Theorem~\ref{thm2sp} by adding long strings of~0's and long strings of~1's (in
alternation) after satisfying each requirement. Then we can make the
density go arbitrarily low and then arbitrarily high infinitely
often.
\end{proof}

It is an immediate consequence of (the proof of) Theorem~\ref{thm:nolimitBPS},
that there are $\BPS$ random reals which are not $\PSPACE$ random
in the sense of Lutz~\cite{Lutz92}.
To see this, note that part~4 of example~2.10 of
Ambos-Spies and Mayordomo~\cite{AM97} shows that
$\PSPACE$ random reals cannot have ``exponential gaps''.
For the same reason, it follows that there are
$\BPS$ random reals which are not PSR1-/PSR2-random in the sense
of Ko~\cite{Ko86}, as Ko showed that the relative frequencies of
$0$'s and~$1$'s in a PSR1-/PSR2-random real quickly converge to~$1/2$.

The constructions of Section~\ref{sec:relativeBP}
carry over straightforwardly to $\BPS$ randomness.  Since we
only use relative computation relative to 0/1 valued functions,
the definitions of the various
forms of relative $\BPS$ randomness are completely straightforward.
It is then easy to prove the
equivalence of the the Kolmogorov definition, the \ml definition, and the martingale
definition of relative randomness by the same constructions as
used for Theorems \ref{thm:firstequivsBPS} and~\ref{thm:martingaleEquivBPS}.
We leave the details to the reader.  One direction of van Lambalgen's theorem
also follows by the same proof as Theorem~\ref{thm:LambalgenDirectionOne}.

\begin{theorem} \label{thm:lambalgenDiectionOneBPS}
For any reals $A$ and~$B$,
if $A$ is $\BPS$ random relative to~$B$ and
$B$~is $\BPS$ random relative to~$A$, then
$A\oplus B$ is $\BPS$ random.
\end{theorem}
On the other hand, the sets $A$ and~$B$ constructed in the
proof of Theorem~\ref{thm:Lambalgen:CounterExample} are
certainly $\BPS$ random relative to each other, and clearly
$A\oplus B$ is not $\BPS$ random.  Thus,
\begin{theorem}\label{thm:Lambalgen:CounterExampleBPS}
There are reals $A$ and~$B$ such that $A$~is $\BPS$ (even, $\BP$) random relative to~$B$,
and $B$~is $\BPS$ (even, $\BP$) random relative to~$A$, but $A\oplus B$ is
not $\BPS$ random.
\end{theorem}

\section{Conclusions and future research}

In this paper, we defined a robust notion of primitive recursive
and $\PSPACE$ bounded random reals
in that both definitions could be framed
in all three versions of algorithmically
random reals via measure,
Kolmogorov complexity, or martingales. We view the work of
this paper as a possible model for defining algorithmically random
reals relative to several other classes of
sub-computable functions. In future work, we will
define similar notions of bounded random reals
for other classes of sub-computable functions such
as elementary, on-line, or exponential space.

Lutz \cite{Lutz90,Lutz92} proved a wide range of
properties of $\PSPACE$ randomness.
It would be interesting to understand
which of the properties of $\PSPACE$ randomness
established in~\cite{Lutz90,Lutz92} apply also to
$\BPS$ randomness.
Lutz showed that
$\Ptime$ has measure~0 in $\EXPTIME$.
This means that
there is a single $\EXPTIME$ martingale
which succeeds on every set in~$\Ptime$.
It is an interesting question whether there can be a $\BPS$
martingale which succeeds
on every set in~$\Ptime$.

Kihara and Miyabe \cite{KM13} defined a notion of uniform
relativization for Kurtz randomness  for which
van Lambalgen's theorem holds.  It would be interesting
to investigate whether their uniform relativization can
be adapted to the setting of $\BP$ and $\BPS$ randomness.

A theory of algorithmic randomness for trees and
effectively closed sets was developed in a series of papers by
Barmpalias, Brodhead, Cenzer, et al.~\cite{BCDW07,BCRW08}.
One can adapt our definitions of primitive recursive
bounded randomness to define similar notions
of bounded random trees and
effectively closed sets for various classes of
sub-computable functions. This will appear in future papers.
\bigskip

\noindent
{\bf Acknowledgements.}
We thank the anonymous referees for this paper,
as well the referees for~\cite{CenzerRemmel:BPrandomness},
for useful comments and suggestions.

\bibliographystyle{plain}
\bibliography{random}

\end{document}